\documentclass[12pt]{article}
\usepackage{fullpage}

\usepackage[english]{babel}
\usepackage{enumerate}
\usepackage[pdfborder=0]{hyperref}
\usepackage{amsmath,amsthm,amsfonts,amssymb}
\usepackage[nottoc, notlof, notlot]{tocbibind}
\usepackage{bm}
\usepackage{color}
\usepackage{alltt}
\usepackage{fixltx2e} 

\title{
\vspace*{-2.5cm}
Faster Algorithms for Multivariate Interpolation with Multiplicities and Simultaneous Polynomial Approximations
\footnote{The material in this paper was presented in part at the 10th Asian Symposium on Computer Mathematics
(ASCM), Beijing, China, October 2012 and at SIAM Conference on Applied Algebraic Geometry, Fort Collins, Colorado, USA, August 2013.}}
\author{
Muhammad F. I. Chowdhury\footnote{Computer Science Department, University of Western Ontario, London ON, Canada.}\,, 
Claude-Pierre Jeannerod\footnote{Inria, Laboratoire LIP (CNRS, ENS de Lyon, Inria, UCBL), Universit\'e de Lyon, France.}\,, \\ 
Vincent Neiger\footnote{ENS de Lyon, Laboratoire LIP (CNRS, ENS de Lyon, Inria, UCBL), Universit\'e de Lyon, France; 
						Computer Science Department, University of Western Ontario, London ON, Canada.}\,, 
\'Eric Schost\footnote{Computer Science Department, University of Western Ontario, London ON, Canada.}\,, 
Gilles Villard\footnote{CNRS, Laboratoire LIP (CNRS, ENS de Lyon, Inria, UCBL), Universit\'e de Lyon, France.}
}
\date{\small\today}

\newcommand{\M}{\ensuremath{\mathsf{M}} }
\newcommand{\K}{\mathbb{K}}
\newcommand{\extK}{\mathbb{L}}
\newcommand{\Z}{\mathbb{Z}}
\newcommand{\Zp}{\mathbb{Z}_{>0}}
\newcommand{\Znn}{\mathbb{Z}_{\geqslant 0}}

\newcommand{\A}{\mathbb{A}}
\newcommand{\ZZ}{\mathcal{Z}}
\newcommand{\OO}{\mathcal{O}}
\newcommand{\PP}{\mathcal{P}}
\DeclareBoldMathCommand{\bi}{i}
\DeclareBoldMathCommand{\bii}{\overline{\imath}}
\DeclareBoldMathCommand{\bj}{j}
\DeclareBoldMathCommand{\bjj}{\overline{\jmath}}
\DeclareBoldMathCommand{\bJ}{J}
\DeclareBoldMathCommand{\bk}{k}
\DeclareBoldMathCommand{\bm}{m}
\DeclareBoldMathCommand{\ba}{a}
\DeclareBoldMathCommand{\bY}{Y}
\DeclareBoldMathCommand{\by}{y}
\DeclareBoldMathCommand{\bL}{L}
\DeclareBoldMathCommand{\bR}{R}
\DeclareBoldMathCommand{\bF}{F}
\DeclareBoldMathCommand{\bN}{N}
\DeclareBoldMathCommand{\bM}{M}
\DeclareBoldMathCommand{\bQ}{Q}
\DeclareBoldMathCommand{\bzero}{0}
\newcommand{\wdeg}{\mathrm{wdeg}}
\newcommand{\sbi}{\ensuremath{{|\bi|}} }
\newcommand{\sbj}{\ensuremath{{|\bj|}} }
\newcommand{\sbJ}{\ensuremath{{|\bJ|}} }
\renewcommand{\le}{\leqslant}
\renewcommand{\ge}{\geqslant}
\newcommand{\softO}[1]{\mathcal{O}\tilde{~}(#1)}
\newcommand{\hyp}[1]{\textbf{H\textsubscript{#1}}}

\theoremstyle{definition}
\newtheorem{dfn}{Definition}
\newtheorem{pbm}{Problem}
\newtheorem{algo}{Algorithm}
\theoremstyle{plain}
\newtheorem{thm}[dfn]{Theorem}
\newtheorem{cor}[dfn]{Corollary}
\newtheorem{prop}[dfn]{Proposition}
\newtheorem{lem}[dfn]{Lemma}

\begin{document}

\maketitle
\vspace*{-1cm}
\begin{abstract}
  The interpolation step in the Guruswami-Sudan algorithm 
  is a bivariate interpolation problem with multiplicities
  commonly solved in the literature using either structured linear 
  algebra or basis reduction of polynomial lattices. 
  This problem has been extended to three or more variables;
  for this generalization, all fast algorithms proposed so far 
  rely on the lattice approach.
  In this paper, we reduce this multivariate interpolation problem 
  to a problem of simultaneous polynomial approximations, which we 
  solve using fast structured linear algebra.
  This improves the best
  known complexity bounds for the interpolation step of the
  list-decoding of Reed-Solomon codes, Parvaresh-Vardy codes, and folded
  Reed-Solomon codes. In particular, for Reed-Solomon list-decoding
  with re-encoding, our approach has complexity 
  $\mathcal{O}\tilde{~}(\ell^{\omega-1}m^2(n-k))$,
  where $\ell,m,n,k$ are the list size, the multiplicity, the number
  of sample points and the dimension of the code, and $\omega$ is the 
  exponent of linear algebra; this accelerates the previously fastest
  known algorithm by a factor of $\ell / m$.
\end{abstract}

\section{Introduction}
\label{sec:introduction}

\paragraph{Problems.}
In this paper, we consider a multivariate interpolation problem 
with multiplicities and degree constraints 
(Problem~\ref{pbm:multivariate_interpolation} below) which
originates from coding theory. In what follows, $\K$ is our base field
and, in the coding theory context, $s, \ell, n, b$ are
respectively known as the \emph{number of variables}, \emph{list
size}, \emph{code length}, and as an {\em agreement parameter}. 
The parameters $m_1,\ldots,m_n$ are known as \emph{multiplicities}
associated with each of the $n$ points; furthermore, the $s$ 
variables are associated with some \emph{weights} $k_1,\ldots,k_s$.
In the application to list-decoding of Reed-Solomon codes, we have $s=1$,
all the multiplicities are equal to a same value $m$, $n-b/m$ is an 
upper bound on the number of errors allowed on a received word,
and the weight $k := k_1$ is such that $k+1$ is 
the \emph{dimension} of the code.
Further details concerning the applications of our results 
to list-decoding and soft-decoding of Reed-Solomon codes 
are given in Section~\ref{sec:rs_codes}.

We stress that here we do not address the issue of choosing the
parameters $s,\ell,m_1,\ldots,m_n$ with respect to $n,b,k_1,\ldots,k_s$, 
as is often done:
in our context, these are all input parameters. Similarly, although we
will mention them, we do not make some usual assumptions on these
parameters; in particular, we do not make any assumption that 
ensures that our problem admits a solution: the algorithm will
detect whether no solution exists.

Here and hereafter, $\Z$ is the set of integers, $\Znn$ the set
of nonnegative integers, and $\Zp$ the set of positive integers.
Besides, $\deg_{Y_1,\ldots,Y_s}$ denotes the total degree with respect 
to the variables $Y_1,\ldots,Y_s$, and $\wdeg_{k_1,\ldots,k_s}$ denotes 
the weighted-degree with respect to weights $k_1,\ldots,k_s \in \Z$ 
on variables $Y_1,\ldots,Y_s$, respectively; 
that is, for a polynomial $Q =
\sum_{(j_1,\ldots,j_s)} Q_{j_1,\ldots,j_s}(X) Y_1^{j_1} \cdots Y_s^{j_s}$,
\[\wdeg_{k_1,\ldots,k_s}(Q) \;=\; \max_{j_1,\ldots,j_s} \big(\deg(Q_{j_1,\ldots,j_s}) 
\:+\: j_1k_1 + \cdots + j_sk_s\big).\]

\begin{center}
\fbox{\begin{minipage}{14.8cm}
\begin{pbm} \textsc{MultivariateInterpolation}
  \label{pbm:multivariate_interpolation}

\medskip
  \noindent
  \emph{Input:} 
  $s,\ell,n,m_1,\ldots,m_n$ in $\Zp$, 
  $b,k_1,\ldots,k_s$ in $\Z$ and  points \\
  \hspace*{1.15cm} $\{(x_i,y_{i,1},\ldots,y_{i,s})\}_{1\le i\le n}$ 
  in $\K^{s+1}$ with the $x_i$ pairwise distinct.

  \medskip
  \noindent
  \emph{Output:}
  a polynomial $Q$ in $\K[X,Y_1,\ldots,Y_s]$ such that
  \begin{enumerate}[{\em (i)}]
  \item\label{cond:i} $Q$ is nonzero,
  \item\label{cond:ii} $\deg_{Y_1,\ldots,Y_s}(Q) \,\le  \, \ell$,
  \item\label{cond:iii} $\wdeg_{k_1,\ldots,k_s}(Q) \,<\, b$,
  \item\label{cond:iv} for $1 \le i \le n$,
    $Q(x_i,y_{i,1},\ldots,y_{i,s}) = 0$ with multiplicity at least $m_i$.
  \end{enumerate}
\end{pbm}
\end{minipage}}
\end{center}

\smallskip
We call conditions~\eqref{cond:ii},~\eqref{cond:iii},
and~\eqref{cond:iv} the \emph{list-size} condition, the
\emph{weighted-degree} condition, and the \emph{vanishing} condition,
respectively.  Note that a point
$(x,y_1,\ldots,y_s)$ is a zero of $Q$ of {\em multiplicity at
  least} $m$ if 
the shifted polynomial $Q(X+x, Y_1+y_1, \ldots, Y_s+y_s)$ 
has no monomial of total degree less than $m$;
in characteristic zero or larger than $m$, this is equivalent to 
requiring that all the derivatives of $Q$ of order up to $m-1$ 
vanish at $(x,y_1,\ldots,y_s)$.

By linearizing condition~\eqref{cond:iv} under the assumption that 
conditions~\eqref{cond:ii} and~\eqref{cond:iii} are satisfied, 
it is easily seen that solving Problem~\ref{pbm:multivariate_interpolation}
amounts to computing a nonzero solution to an $M\times N$ 
homogeneous linear system over $\K$.
Here, the number $M$ of equations derives from condition~\eqref{cond:iv} and
thus depends on $s$, $n$, $m_1,\ldots,m_n$, while the number $N$ of unknowns derives from
conditions~\eqref{cond:ii} and~\eqref{cond:iii} and thus depends on 
$s$, $\ell$, $b$, $k_1,\ldots,k_s$.
It is customary to assume $M < N$
in order to guarantee the existence of a nonzero solution; 
however, as said above, we do not make this assumption, 
since our algorithms do not require it.

Problem~\ref{pbm:multivariate_interpolation} is a generalization
of the interpolation step of the Guruswami-Sudan 
algorithm~\cite{Sudan97,GurSud99}
to $s$ variables $Y_1,\dots,Y_s$, distinct multiplicities, and distinct 
weights. The multivariate case $s > 1$ occurs for instance in Parvaresh-Vardy 
codes~\cite{PaVa05} or folded Reed-Solomon codes~\cite{GuRu08}.
Distinct multiplicities occur for instance in the interpolation step 
in soft-decoding of Reed-Solomon codes~\cite{KoeVar03a}.
We note that this last problem is different from our context since 
the $x_i$ are not necessarily pairwise distinct; we briefly explain 
in Section~\ref{subsec:softdec} how to deal with this case.

Our solution to Problem~\ref{pbm:multivariate_interpolation} relies on
a reduction to a simultaneous approximation problem (Problem~\ref{pbm:poly_approx} below)
which generalizes Pad\'e and Hermite-Pad\'e approximation.

\begin{center}
\fbox{\begin{minipage}{14.8cm}
\begin{pbm} \textsc{SimultaneousPolynomialApproximations} 
\label{pbm:poly_approx}

\medskip
\noindent
\emph{Input:} 
 $\mu,\ \nu$, $M'_0,\dots,M'_{\mu-1}$,
 $N'_0,\dots,N'_{\nu-1}$ in $\Zp$ and polynomial tuples \\
 \hspace*{1.15cm} $\{(P_i,F_{i,0},\ldots,F_{i,\nu-1})\}_{0 \le i < \mu}$ in $\K[X]^{\nu+1}$
 such that for all $i$, $P_i$ is monic of \\
 \hspace*{1.15cm} degree $M'_i$ and $\deg(F_{i,j}) < M'_i$ for all $j$. 

  \medskip
  \noindent

\medskip\noindent
\emph{Output:}
polynomials $Q_0,\ldots,Q_{\nu-1}$ in $\K[X]$ satisfying the following conditions:

\begin{enumerate}[$(a)$]
\item\label{cond:a} the $Q_j$ are not all zero,
\item\label{cond:b} for $0 \le j<\nu$, $\deg(Q_j) < N'_j,$
\item\label{cond:c} for $0 \le i<\mu$, $\sum_{0 \le j <\nu} F_{i,j} Q_j =  0 \bmod{P_i}$.
\end{enumerate}
\end{pbm}
\end{minipage}}
\end{center}
\smallskip

\paragraph{Main complexity results and applications.}
We first show in Section~\ref{sec:red} how to reduce 
Problem~\ref{pbm:multivariate_interpolation} 
to Problem~\ref{pbm:poly_approx} 
efficiently via a generalization of the techniques introduced by Zeh, Gentner, and Augot~\cite{ZeGeAu11} 
and Zeh~\cite[Section 5.1.1]{Zeh13} for, respectively, the list-decoding and soft-decoding of Reed-Solomon codes.

Then, in Section~\ref{sec:solutions} we present two algorithms for solving Problem~\ref{pbm:poly_approx}. 
Each of them involves a linearization of the univariate equations~\eqref{cond:iiip} into a
specific homogeneous linear system over $\K$; if we define
$$M' = \sum_{0 \le i < \mu} M'_i \quad\text{and}\quad N' = \sum_{0 \le
  j < \nu} N'_j,$$ then both systems have $M'$ equations in $N'$
unknowns. (As for our first problem, we need not assume that $M' < N'$.)
Furthermore, the structure of these systems allows us to solve them 
efficiently using the algorithm of Bostan, Jeannerod, and Schost in~\cite{BoJeSc08}. 

Our first algorithm,
detailed in Section~\ref{subsec:key_equations_version}, 
solves Problem~\ref{pbm:poly_approx} by following the
derivation of so-called extended key equations (EKE), initially
introduced for the particular case of Problem~\ref{pbm:multivariate_interpolation}
by Roth and Ruckenstein~\cite{RotRuc00} when $s=m=1$ and then by Zeh, 
Gentner, and Augot~\cite{ZeGeAu11} when $s=1$ and $m\ge 1$;
the matrix of
the system is mosaic-Hankel. In our second algorithm, detailed in
Section~\ref{subsec:direct}, the linear system is more directly
obtained from condition~\eqref{cond:iiip}, without resorting to EKEs,
and has Toeplitz-like structure.

Both points of view lead to the same complexity result, stated in
Theorem~\ref{thm:poly_approx_complexity} below, which says that
Problem~\ref{pbm:poly_approx} can be solved in time quasi-linear in
$M'$, multiplied by a subquadratic term in $\rho=\max(\mu,\nu)$. In
the following theorems, and the rest of this paper, the soft-O
notation $\mathcal{O}\tilde{~}(\ )$ indicates that we omit
polylogarithmic terms. The exponent $\omega$ is so that we can
multiply $n \times n$ matrices in $\mathcal{O} ( n^{\omega})$ ring
operations on any ring, the best known bound being $\omega
< 2.38$~\cite{CoWi90,St10,Williams12,LeGall14}. Finally, the function $\M$
is a {\em multiplication time} function for~$\K[X]$: $\M$ is such that
polynomials of degree at most $d$ in $\K[X]$ can be multiplied in
$\M(d)$ operations in $\K$, and satisfies the
super-linearity properties of~\cite[Ch.~8]{vzGathen03}. 
It follows from the algorithm of Cantor and Kaltofen~\cite{CaKa91}
that $\M(d)$ can be taken in $\mathcal{O}(d \log(d)\log\log(d)) \subseteq \mathcal{O}\tilde{~}(d)$.

Combining Theorem~\ref{thm:poly_approx_complexity} below with
the above-mentioned reduction from 
Problem~\ref{pbm:multivariate_interpolation} to 
Problem~\ref{pbm:poly_approx},
we immediately deduce the following cost bound for
Problem~\ref{pbm:multivariate_interpolation}.
\begin{thm}
  \label{thm:complexity}
  Let 
  \[ \Gamma = \big\{(j_1,\ldots,j_s)\in\Znn^s \,\,\mid\,\, j_1+\cdots+j_s\le\ell 
  \quad\text{and}\quad j_1k_1+\cdots +j_sk_s<b \big\}, \]
  and let $m = \max_{1\le i \le n} m_i$, $\varrho=\max\big(|\Gamma|,\binom{s+m-1}{s}\big)$ 
  and $M = \sum_{1\le i\le n} \binom{s+m_i}{s+1}$.
  There exists a probabilistic algorithm that either computes a
  solution to Problem~\ref{pbm:multivariate_interpolation}, or
  determines that none exists, using
  \[ \mathcal{O}\big(\varrho^{\omega-1}\M(M)\log(M)^2\big) \subseteq
  \mathcal{O}\tilde{~}(\varrho^{\omega-1} M)\] operations in $\K$.
  This can be achieved using Algorithm~\ref{algo:reduction} 
  in Section~\ref{sec:red} followed by 
  Algorithm~\ref{algo:interpolation_eke} 
  or~\ref{algo:interpolation_direct} in Section~\ref{sec:solutions}.
  These algorithms choose $\mathcal{O}(M)$ elements in $\K$; if these
  elements are chosen uniformly at random in a set $S \subseteq \K$ of
  cardinality at least $6(M+1)^2$, then the probability of~success is at
  least $1/2$.
\end{thm}

We will often refer to the two following assumptions 
on the input parameters:
\begin{itemize}
	\item[] \hyp{1}: $m \;\le\; \ell$, 
	\item[] \hyp{2}: $b>0$ and $b > \ell \cdot \max_{1\le j\le s} k_j$.
\end{itemize}
Regarding \hyp{1}, we prove in Appendix~\ref{app:H1} that 
the case $m > \ell$ can be reduced to the case $m = \ell$,
so that this assumption can be made without loss of generality.
Besides, it is easily verified that \hyp{2} is equivalent to having
$\Gamma = \{(j_1,\ldots,j_s)\in\Znn^s \,\,\mid\,\, j_1+\cdots+j_s\le\ell\}$;
when $k_j>0$ for some $j$, \hyp{2} means that we do not take 
$\ell$ uselessly large. Then, assuming \hyp{1} and \hyp{2}, we have 
$\varrho = |\Gamma| = \binom{s+\ell}{s}$.

As we will show in Section~\ref{sec:rs_codes}, in the context of the
list-decoding of Reed-Solomon codes, applications of Theorem~\ref{thm:complexity} 
include the 
interpolation step of the Guruswami-Sudan algorithm~\cite{GurSud99} in 
$\softO{\ell^{\omega-1}m_{\mathrm{GS}}^2n}$ operations
and the interpolation step of the Wu algorithm~\cite{Wu08} in
$\softO{\ell^{\omega-1}m_{\mathrm{Wu}}^2n}$ operations,
where $m_{\mathrm{GS}}$ and $m_{\mathrm{Wu}}$ are the respective 
multiplicities used in those algorithms; our result can also be adapted
to the context of soft-decoding~\cite{KoeVar03a}.
Besides, the re-encoding technique of Koetter and Vardy~\cite{KoeVar03b} 
can be used in conjunction with our algorithm in order to 
reduce the cost of the interpolation step of the Guruswami-Sudan 
algorithm to
$\softO{\ell^{\omega-1}m_{\mathrm{GS}}^2(n-k)}$ operations.

In Theorem~\ref{thm:complexity}, the probability analysis is a standard 
consequence of the Zippel-Schwartz lemma; as usual,
the probability of success can be made arbitrarily close 
to one by increasing the size of $S$. 
If the field $\K$ has fewer than $6(M+1)^2$ elements,
then a probability of success at least $1/2$ can still be 
achieved by using a field extension $\extK$ of degree
$d \in \mathcal{O}(\log_{|\K|}(M))$, up to a cost increase 
by a factor in $\mathcal{O}(\M(d) \log(d))$. 

Specifically, one can proceed in three steps.
First, we take $\extK = \K[X]/\langle f \rangle$ with 
$f\in\K[X]$ irreducible of degree $d$;
such an $f$ can be set up using an expected number of 
$\mathcal{O}\tilde{~}(d^2) \subseteq \mathcal{O}(M)$ 
operations in $\K$~\cite[\S14.9]{vzGathen03}.
Then we solve Problem~\ref{pbm:multivariate_interpolation} over $\extK$ 
by means of the algorithm of Theorem~\ref{thm:complexity},
thus using $\mathcal{O}\big(\varrho^{\omega-1} \M(M)\log(M)^2 
\cdot \M(d) \log(d)\big)$ operations in $\K$.
Finally, from this solution over $\extK$ one can deduce 
a solution over $\K$ using $\mathcal{O}(Md)$ operations in $\K$. 
This last point comes from the fact that, as we shall see later 
in the paper, Problem~\ref{pbm:multivariate_interpolation} amounts 
to finding a nonzero vector $u$ over $\K$ such that $Au = 0$ for some 
$M\times (M+1)$ matrix~$A$ over~$\K$: once we have obtained a solution 
$\overline{u}$ over $\extK$, it thus suffices to rewrite it as 
$\overline{u} = \sum_{0\le i < d} u_i X^i \ne 0$ and, 
noting that $Au_i = 0$ for all $i$, to find a nonzero 
$u_i$ in $\mathcal{O}(Md)$ and return it as a solution over $\K$.

Furthermore, since the $x_i$ in Problem~\ref{pbm:multivariate_interpolation} 
are assumed to be pairwise distinct, we have already $|\K| \ge n$ and 
thus we can take $d = \mathcal{O}(\log_n (M))$.
In all the applications to error-correcting codes we consider in this
paper, $M$ is polynomial in $n$ 
so that we can take $d = \mathcal{O}(1)$, and in those cases the cost 
bound in Theorem~\ref{thm:complexity} holds for any field.

As said before, Theorem~\ref{thm:complexity} relies on an efficient
solution to Problem~\ref{thm:poly_approx_complexity}, which we summarize
in the following theorem.
\begin{thm}
  \label{thm:poly_approx_complexity}
  Let $\rho=\max(\mu,\nu)$.
  There exists a probabilistic algorithm that either computes a
  solution to Problem~\ref{pbm:poly_approx}, or determines that none
  exists, using \[ \mathcal{O}\big(\rho^{\omega-1} \M(M')\log(M')^2\big) \subseteq
  \mathcal{O}\tilde{~}(\rho^{\omega-1} M' )\] operations in $\K$.
  Algorithm~\ref{algo:interpolation_eke} and~\ref{algo:interpolation_direct} in 
  Section~\ref{sec:solutions} achieve this result.
  These algorithms both choose $\mathcal{O}(M')$ elements in $\K$; if these
  elements are chosen uniformly at random in a set $S \subseteq \K$ of
  cardinality at least $6(M'+1)^2$, then the probability of success is at
  least $1/2$.
\end{thm}
If $\K$ has fewer than $6(M'+1)^2$ elements, the remarks made after
Theorem~\ref{thm:complexity} still apply here.

\paragraph{Comparison with previous work.}
In the context of coding theory, most previous results 
regarding Problem~\ref{pbm:multivariate_interpolation}
focus on the list-decoding of
Reed-Solomon codes via the Guruswami-Sudan algorithm, in which 
$s=1$ and the assumptions \hyp{1} and \hyp{2} are satisfied
as well as
\begin{enumerate}
	\item[] \hyp{3}: $0\le k<n$ where $k := k_1$,
	\item[] \hyp{4}: $m_1 = \cdots = m_n \;= m$.
\end{enumerate}
The assumption~\hyp{3} corresponds to the coding theory context, 
where $k+1$ is the dimension of the code; then $k+1$ must be 
positive and at most $n$ (the length of the received word).
To support this assumption independently from any 
application context, we show in Appendix~\ref{app:H5} that
if $k\ge n$, then Problem~\ref{pbm:multivariate_interpolation}
has either a trivial solution or no solution at all.

Previous results focus mostly on the Guruswami-Sudan case 
(${s=1,m\ge 1}$) and some of them more specifically on the 
Sudan case (${s=m=1}$); we summarize these results in 
Table~\ref{table:prev}. In some 
cases~\cite{Reinhard03,Alekhnovich05,Bernstein11,CohHen11a}, 
the complexity was not stated quite exactly in our terms 
but the translation is straightforward.

In the second column of that table, we give the cost with respect to 
the interpolation parameters $\ell,m,n$, assuming further
$m = n^{\mathcal{O}(1)}$ and $\ell = n^{\mathcal{O}(1)}$.
The most significant factor in the running time is its
dependency with respect to $n$, with results being either cubic,
quadratic, or quasi-linear. Then, under the assumption $\hyp{1}$,
the second most important parameter is $\ell$, followed by $m$.
In particular, our result in Section~\ref{sec:rs_codes}, 
Corollary~\ref{cor:complexity_gursud} compares favorably to the cost 
$\mathcal{O}\tilde{~}(\ell^\omega mn)$ obtained by Cohn and 
Heninger~\cite{CohHen11a} which was, to our knowledge, 
the best previous bound for this problem.

In the third column, we give the cost with respect to the Reed-Solomon
code parameters $n$ and $k$, using worst-case parameter choices that
are made to ensure the existence of a solution: $m=\mathcal{O}(nk)$ and 
$\ell=\mathcal{O}(n^{3/2}k^{1/2})$ in the Guruswami-Sudan 
case~\cite{GurSud99} and 
$\ell = \mathcal{O}(n^{1/2}k^{-1/2})$ in the Sudan 
case~\cite{Sudan97}. 
With these parameter choices, our algorithms present a
speedup $(n/k)^{1/2}$ over the algorithm in~\cite{CohHen11a}.

\begin{table}[h]
  \centering
  \caption{Comparison of our costs with previous ones for $s=1$}
  \label{table:prev} \bigskip
  \begin{tabular}{|| p{5cm} | c | c ||}
	\hline
	\multicolumn{3}{||c||}{Sudan case ($m=1$)}\\
	\hline\hline
    Sudan~\cite{Sudan97} & $\mathcal{O}(n^3)$ & $\mathcal{O}(n^3)$\\
    Roth-Ruckenstein~\cite{RotRuc00} & $\mathcal{O}(\ell n^2)$ & $\mathcal{O}(n^{2+1/2}k^{-1/2})$ \\
    Olshevsky-Shokrollahi~\cite{OlsSho99} & $\mathcal{O}(\ell n^2)$ & $\mathcal{O}(n^{2+1/2}k^{-1/2})$ \\
	\emph{This paper} & $\mathcal{O}(\ell^{\omega-1} \M(n) \log(n)^2)$ & $\mathcal{O}\tilde{~}(n^{\omega/2+1/2} k^{1/2-\omega/2})$\\
	\hline\hline
	\multicolumn{3}{||c||}{Guruswami-Sudan case ($m\geqslant 1$)} \\
	\hline\hline
    Guruswami-Sudan~\cite{GurSud99} & $\mathcal{O}(m^6 n^3)$ & $\mathcal{O}(n^9k^6)$\\
    Olshevsky-Shokrollahi~\cite{OlsSho99} & $\mathcal{O}(\ell m^4 n^2)$ & $\mathcal{O}(n^{7+1/2}k^{4+1/2})$\\
    Zeh-Gentner-Augot~\cite{ZeGeAu11} & $\mathcal{O}(\ell m^4 n^2)$ & $\mathcal{O}(n^{7+1/2}k^{4+1/2})$\\
    K\"otter / McEliece~\cite{Koetter96,McEliece03} & $\mathcal{O}(\ell m^4 n^2)$ & $\mathcal{O}(n^{7+1/2}k^{4+1/2})$\\
    Reinhard~\cite{Reinhard03} & $\mathcal{O}(\ell^3 m^2 n^2)$ & $\mathcal{O}(n^{8+1/2}k^{3+1/2})$\\
    Lee-O'Sullivan~\cite{LeeOSul08}  & $\mathcal{O}(\ell^4 m n^2)$ & $\mathcal{O}(n^9k^3)$\\
    Trifonov~\cite{Trifonov10} (heuristic)  & $\mathcal{O}(m^3 n^2)$ & $\mathcal{O}(n^5k^3)$\\
    Alekhnovich~\cite{Alekhnovich05} & $\mathcal{O}(\ell^4 m^4 \M(n)\log(n))$ & $\mathcal{O}\tilde{~}(n^{11} k^6)$\\
    Beelen-Brander~\cite{BeeBra10} & $\mathcal{O}(\ell^3 \M(\ell m n)\log(n))$ & $\mathcal{O}\tilde{~}(n^8k^3)$\\
    Bernstein~\cite{Bernstein11} & $\mathcal{O}(\ell^\omega\M(\ell n) \log(n))$ & $\mathcal{O}\tilde{~}(n^{3\omega/2+5/2}k^{\omega/2+1/2})$\\
    Cohn-Heninger~\cite{CohHen11a} & $\mathcal{O}(\ell^\omega \M(m n) \log(n))$ & $\mathcal{O}\tilde{~}(n^{3\omega/2+2}k^{\omega/2+1})$\\ 
	\emph{This paper} & $\mathcal{O}\big(\ell^{\omega-1} \M(m^2n) \log(n)^2 \big)$ & $\mathcal{O}\tilde{~}(n^{3\omega/2+3/2}k^{\omega/2+3/2})$
\\
	\hline
  \end{tabular}
\end{table}

\medskip
Most previous algorithms rely on linear algebra, either over $\K$ or
over $\K[X]$. When working over $\K$, a natural idea is to rely on
cubic-time general linear system solvers, as in Sudan's and
Guruswami-Sudan's original papers. Several papers also cast the
problem in terms of Gr{\"o}bner basis computation in $\K[X,Y]$,
implicitly or explicitly: the incremental algorithms
of~\cite{Koetter96,NiHo98,McEliece03} are particular cases of the
Buchberger-M{\"o}ller algorithm~\cite{MoBu82}, while Alekhnovich's
algorithm~\cite{Alekhnovich05} is a divide-and-conquer change of
term order for bivariate ideals.

Yet another line of work~\cite{RotRuc00,ZeGeAu11} uses Feng and Tzeng's
linear system solver~\cite{FeTz91}, combined with a reformulation in
terms of syndromes and key equations. We will use (and generalize to
the case $s>1$) some of these results in
Section~\ref{subsec:key_equations_version}, but we will rely on the
structured linear system solver of~\cite{BoJeSc08} in order to prove
our main results. Prior to our work, Olshevsky and Shokrollahi also
used structured linear algebra techniques~\cite{OlsSho99}, but it is
unclear to us whether their encoding of the problem could lead to
similar results as ours.

As said above, another approach rephrases the problem of computing $Q$
in terms of polynomial matrix computations, that is, as linear algebra
over $\K[X]$. Starting from known generators of the finitely generated
$\K[X]$-module (or polynomial lattice) formed by solutions to 
Problem~\ref{pbm:multivariate_interpolation},
the algorithms in~\cite{Reinhard03,LeeOSul08,PB08,BeeBra10,KB10,Bernstein11,CohHen11a} 
compute a Gr{\"o}bner basis of this module (or a reduced lattice basis), 
in order to find a short vector therein.
To achieve quasi-linear time in $n$, the algorithms
in~\cite{BeeBra10,KB10} use a basis reduction subroutine due to
Alekhnovich~\cite{Alekhnovich05}, while those
in~\cite{Bernstein11,CohHen11a} rely on a faster, randomized
algorithm due to Giorgi, Jeannerod, and Villard~\cite{GiJeVi03}.  

This approach based on the computation of a reduced lattice basis
was in particular the basis of the extensions to the multivariate 
case $s > 1$ in~\cite{PB08,KB10,CohHen12}.
In the multivariate case as well, the result in Theorem~\ref{thm:complexity} 
improves on the best previously known bounds~\cite{PB08,KB10,CohHen12};
we detail those bounds and we prove this claim in 
Appendix~\ref{app:lattice_comparison}.
In~\cite{GabRua06b}, the authors solve a problem similar to 
Problem~\ref{pbm:multivariate_interpolation} except that they do 
not assume that the $x_i$ are distinct. 
For simple roots and under some genericity
assumption on the points
$\{(x_i,y_{i,1},\ldots,y_{i,s})\}_{1\leqslant i\leqslant n}$,
this algorithm uses $O(n^{2+1/s})$ operations to compute a polynomial $Q$
which satisfies \eqref{cond:i}, \eqref{cond:iii}, \eqref{cond:iv}
with $m=1$.
However, the complexity analysis is not clear to us in the general 
case with multiple roots ($m> 1$).

\medskip
Regarding Problem~\ref{pbm:poly_approx}, several particular cases of 
it are well-known. When all $P_i$ are of the form $X^{M_i'}$, this 
problem becomes known as a simultaneous Hermite-Pad\'e approximation
problem or vector Hermite-Pad\'e approximation
problem~\cite{BecLab94,Storjohann06}. The case $\mu=1$, with $P_1$
being given through its roots (and their multiplicities) is known as
the M-Pad\'e problem~\cite{Beckermann92}.
To our knowledge, the only previous work on 
Problem~\ref{pbm:poly_approx} in its full generality is by Nielsen 
in~\cite[Chapter 2]{Nielsen13}. Nielsen solves the problem
by building an ad-hoc polynomial lattice, which has dimension 
$\mu+\nu$ and degree $\max_{i<\mu} M'_i$, and finding a short vector
therein. Using the algorithm in~\cite{GiJeVi03}, the overall cost 
bound for this approach is 
$\mathcal{O}\tilde{~}((\mu+\nu)^\omega (\max_{i<\mu} M'_i))$,
to which our cost bound 
$\mathcal{O}\tilde{~}(\max(\mu,\nu)^{\omega-1} (\sum_{i<\mu} M'_i))$ 
from Theorem~\ref{thm:poly_approx_complexity} compares favorably.

\paragraph{Outline of the paper.} 
First, we show in Section~\ref{sec:red} how to reduce
Problem~\ref{pbm:multivariate_interpolation} to
Problem~\ref{pbm:poly_approx}; this reduction is essentially based on 
Lemma~\ref{lem:univariate_reformulation}, which extends to the multivariate
case $s> 1$ the results in~\cite{ZeGeAu11,Zeh13}.
Then, after a reminder on algorithms for structured linear systems in 
Section~\ref{subsec:structured}, we give two
algorithms that both prove Theorem~\ref{thm:poly_approx_complexity}, in
Sections~\ref{subsec:key_equations_version} and~\ref{subsec:direct},
respectively.
The linearization in the first algorithm extends
the derivation of extended key equations presented in~\cite{ZeGeAu11} 
to the more general context of Problem~\ref{pbm:poly_approx}, 
 ending up with a mosaic-Hankel system.
The second algorithm gives an alternative approach,
in which the linearization is more straightforward and the structure 
of the matrix of the system is Toeplitz-like.
We conclude in Section~\ref{sec:rs_codes} by presenting several
applications to the list-decoding of Reed-Solomon codes, namely the
Guruswami-Sudan algorithm, the re-encoding technique and the
Wu algorithm, and by sketching how to adapt our approach to 
the soft-decoding of Reed-Solomon codes.
Readers who are mainly interested in those applications may
skip Section~\ref{sec:solutions}, which contains the proofs of
Theorems~\ref{thm:complexity} and~\ref{thm:poly_approx_complexity},
and go directly to Section~\ref{sec:rs_codes}.


\section{Reducing Problem~\ref{pbm:multivariate_interpolation} to Problem~\ref{pbm:poly_approx}}
\label{sec:red}

In this section, we show how instances of
Problem~\ref{pbm:multivariate_interpolation} can be reduced to
instances of Problem~\ref{pbm:poly_approx}; Algorithm~\ref{algo:reduction}
below gives an overview of this reduction.
The main technical ingredient, stated in 
Lemma~\ref{lem:univariate_reformulation} below,
generalizes to any $s\ge 1$ and (possibly) distinct multiplicities 
the result given for $s=1$ by Zeh, Gentner, and Augot 
in~\cite[Proposition 3]{ZeGeAu11}. To prove it, we use the
same steps as in~\cite{ZeGeAu11}; we rely on the notion of Hasse
derivatives, which allows us to write Taylor expansions in positive
characteristic (see Hasse~\cite{Hasse1936} or
Roth~\cite[pp.~87, 276]{Roth07}).

For simplicity, in the rest of this paper 
we will use boldface letters to denote $s$-tuples of objects:
$\bY = (Y_1,\ldots,Y_s)$, $\bk = (k_1,\ldots,k_s)$, etc.
In the special case of $s$-tuples of integers, we also write $|\bk| = k_1+\cdots+k_s$,
and comparison and addition of multi-indices in $\Znn^s$ are defined componentwise.  
For example, writing $\bi \le \bj$ is equivalent to ($i_1 \le j_1$ and $\ldots$ 
and $i_s \leqslant j_s)$, and $\bi - \bj$ denotes $(i_1-j_1,\ldots,i_s-j_s)$.  
If $\by = (y_1,\ldots,y_s)$ is in $\K[X]^s$ and $\bi = (i_1,\ldots,i_s)$ is in $\Znn^s$,
then $\bY - \by = Y_1-y_1,\ldots,Y_s-y_s$ and $\bY^\bi = Y_1^{i_1}\cdots Y_s^{i_s}$.
Finally, for products of binomial coefficients, we shall write
\[
{\bj \choose \bi} = {j_1 \choose i_1} \cdots {j_s \choose i_s}.
\]
Note that this integer 
is zero when $\bi \not\le \bj$.

If $\A$ is any commutative ring with unity and $\A[\bY]$ denotes the
ring of polynomials in $Y_1,\ldots,Y_s$ over $\A$, then for a
polynomial $P(\bY) = \sum_\bj P_\bj \bY^\bj$ in $\A[\bY]$ and a
multi-index $\bi$ in~$\Znn^s$, the {\em order-$\bi$ Hasse derivative} of
$P$ is the polynomial $P^{[\bi]}$ in $\A[\bY]$ defined by
\[
P^{[\bi]} = \sum_{\bj \ge \bi} {\bj\choose\bi} P_\bj \bY^{\bj-\bi}.
\]
The Hasse derivative satisfies the following property (Taylor expansion):
for all $\ba$ in $\A^s$, 
\[
P(\bY) = \sum_\bi P^{[\bi]}(\ba) (\bY-\ba)^\bi.
\]
The next lemma shows how Hasse derivatives help rephrase the
vanishing condition \eqref{cond:iv} of
Problem~\ref{pbm:multivariate_interpolation} for one of the points
$\{(x_r,\by_r)\}_{1\le r\le n}$. 
\begin{lem} \label{lem:univariate_reformulation_point}
	Let $(x,y_1,\ldots,y_s)$ be a point in $\K^{s+1}$ and
	$\bR = (R_1,\ldots,R_s)$ in $\K[X]^s$ be such that $R_j(x) = y_j$
	for $1\le j\le s$.
	Then, for any polynomial $Q$ in $\K[X,\bY]$, $Q(x,\by) = 0$
	with multiplicity at least $m$ if and only if
	for all $\bi$ in $\Znn^s$ such that $\sbi < m$, 
  \[
  Q^{[\bi]}(X,\bR) = 0 \bmod (X-x)^{m-\sbi}.
  \]
\end{lem}
\begin{proof}
Up to a shift, one can assume that the point is 
$(x,y_1,\ldots,y_s) = (0,\bzero)$; 
in other words, it suffices to show that for $\bR(0) = \bzero \in \K^s$, we
have $Q(0,\bzero) = 0$ with multiplicity at least $m$ if and only if, for all
$\bi$ in $\Znn^s$ such that $\sbi < m$, $X^{m-\sbi}$ divides
$Q^{[\bi]}(X,\bR)$.

Assume first that $(0,\bzero)\in\K^{s+1}$ is a root of $Q$ of multiplicity at
least $m$.  Then, $Q(X,\bY) = \sum_\bj Q_\bj \bY^\bj$ has only
monomials of total degree at least $m$, so that for $\bj \ge \bi$,
each nonzero $Q_\bj\bY^{\bj-\bi}$ has only monomials of total degree
at least $m-\sbi$.  Now, $\bR(0) = \bzero\in\K^s$ implies that $X$
divides each component of $\bR$.  Consequently, $X^{m-\sbi}$ divides
$Q_\bj \bR^{\bj-\bi}$ for each $\bj\ge\bi$, and thus
$Q^{[\bi]}(X,\bR)$ as well.

Conversely, let us assume that for all $\bi$ in $\Znn^s$ such that $\sbi
< m$, $X^{m-\sbi}$ divides $Q^{[\bi]}(X,\bR)$, and show that $Q$ has
no monomial of total degree less than $m$.  Writing the Taylor
expansion of $Q$ with $\A = \K[X]$ and $\ba = \bR$, we obtain
\[
Q(X,\bY) =  \sum_\bi Q^{[\bi]}(X,\bR) (\bY-\bR)^\bi.
\]
Each component of $\bR$ being a multiple of $X$, we deduce that for
the multi-indices $\bi$ such that $\sbi \ge m$ every nonzero monomial
in $Q^{[\bi]}(X,\bR) (\bY-\bR)^\bi$ has total degree at least $m$.
Using our assumption, the same conclusion follows for the
multi-indices such that $\sbi < m$.
\end{proof}

Thus, for each of the points $\{(x_r,\by_r)\}_{1\le r\le n}$ in
Problem~\ref{pbm:multivariate_interpolation}, such a rewriting of 
the vanishing condition~\eqref{cond:iv} for this point holds.
Now intervenes the fact that the $x_i$ are distinct: the polynomials
$(X-x_a)^\alpha$ and $(X-x_b)^\beta$ are coprime for $a\neq b$, so that 
simultaneous divisibility by both those polynomials is equivalent to 
divisibility by their product $(X-x_a)^\alpha \cdot (X-x_b)^\beta$.
Using the $s$-tuple $\bR=(R_1,\dots,R_s) \in \K[X]^s$ 
of Lagrange interpolation polynomials, defined by the conditions
\begin{equation}
\label{eqn:def_lagrange}
\deg(R_j) < n \quad \text{and}\quad R_j(x_i) = y_{i,j}
\end{equation}
for $1 \le i \le n$ and $1\le j \le s$, we can then combine 
Lemma~\ref{lem:univariate_reformulation_point} for all points
so as to rewrite the vanishing condition of 
Problem~\ref{pbm:multivariate_interpolation} 
as a set of modular equations in $\K[X]$ as in 
Lemma~\ref{lem:univariate_reformulation} below.
In this result, we use the notation from 
Problem~\ref{pbm:multivariate_interpolation} 
as well as $m = \max_{1\le r\le n} m_r$.

\begin{lem} \label{lem:univariate_reformulation}
  For any polynomial $Q$ in $\K[X,\bY]$, $Q$ satisfies the
  condition~\eqref{cond:iv} of
  Problem~\ref{pbm:multivariate_interpolation} if and only if
  for all $\bi$ in $\Znn^s$ such that $\sbi < m$, 
  \[
	  Q^{[\bi]}(X,\bR) = 0 \quad\bmod 
	  \prod_{\substack{1\le r\le n: \\ m_r > \sbi}}
	  (X-x_r)^{m_r-\sbi}.
  \]
\end{lem}
\begin{proof}
	This result is easily obtained from 
	Lemma~\ref{lem:univariate_reformulation_point}
	since the $x_r$ are pairwise distinct.
\end{proof}
Note that when all multiplicities are equal, that is, $m=m_1=\cdots=m_n$, 
for every $\sbi$ the modulus takes the simpler form $G^{m-\sbi}$, 
where $G = \prod_{1\le r\le n} (X-x_r)$.

Writing $\bj\cdot\bk = j_1k_1 + \cdots + j_sk_s$, 
recall from the statement of Theorem~\ref{thm:complexity} 
that $\Gamma$ is the set of all $\bj$ in $\Znn^s$ such that $|\bj|\le\ell$ and $\bj\cdot\bk < b$.
Then, defining the positive integers
\[
 N_\bj = b- \bj \cdot \bk
\]
for all $\bj$ in $\Gamma$,
we immediately obtain the following reformulation of the list-size and weighted-degree conditions
of our interpolation problem:
\begin{lem} \label{lem:ii_iii_reformulation}
  For any polynomial $Q$ in $\K[X,\bY]$, $Q$ satisfies the
  conditions~\eqref{cond:ii} and~\eqref{cond:iii} of
  Problem~\ref{pbm:multivariate_interpolation} if and only if
  it has the form 
  \[ 
     Q(X,\bY) = \sum_{ \bj \in \Gamma} Q_\bj(X) \bY^\bj \quad \text{with}\quad
     \text{$\deg(Q_\bj) < N_\bj$.}
  \]
\end{lem}

For $\bi\in\Znn^s$ with $\sbi < m$ and $\bj\in\Gamma$, 
let us now define the polynomials $P_\bi, F_{\bi,\bj}\in\K[X]$ as 
\begin{equation}
\label{eqn:def_pi_fij}
P_\bi = \prod_{\substack{1\le r\le n: \\ m_r > \sbi}} (X-x_r)^{m_r-\sbi}
\qquad\text{and}\qquad F_{\bi,\bj} = {\bj \choose \bi} \bR^{\bj-\bi} \bmod P_\bi.
\end{equation}
It then follows from Lemmas~\ref{lem:univariate_reformulation} and~\ref{lem:ii_iii_reformulation}
that $Q$ in $\K[X,\bY]$ satisfies the conditions~\eqref{cond:ii},~\eqref{cond:iii},~\eqref{cond:iv} 
of Problem~\ref{pbm:multivariate_interpolation} 
if and only if $Q=\sum_{\bj \in \Gamma} Q_\bj \bY^\bj$ for some polynomials $Q_\bj$ in $\K[X]$ such that
\begin{itemize}
 \item $\deg(Q_\bj) < N_\bj$ for all $\bj$ in $\Gamma$,
 \item $\sum_{\bj \in \Gamma} F_{\bi,\bj} Q_\bj = 0 \bmod P_\bi$ for all $|\bi| < m$.
\end{itemize}

Let now $M_\bi$ be the positive integers given by
\[M_\bi = \sum_{1\le r\le n: \; m_r > \sbi} (m_r - \sbi),\]
for all $\sbi < m$.
Since the $P_\bi$ are monic polynomials of degree $M_\bi$
and since $\deg F_{\bi,\bj}< M_\bi$,
the latter conditions express the problem of finding such a $Q$ as an
instance of Problem~\ref{pbm:poly_approx}. In order to make the
reduction completely explicit, define further
\[
 M \;=\; \sum_{|\bi|<m} M_{\bi} \;, 
\]
\[\mu = \binom{s+m-1}{s},\qquad \nu=|\Gamma|, \qquad \varrho=\max(\mu,\nu),\]
choose arbitrary orders on the sets of indices $\{\bi \in \Znn^s \mid
\sbi < m\}$ and $\Gamma$, that is, bijections
\begin{equation}
\label{eqn:orderings}
\phi: \{0,\dots,\mu-1\} \to \{\bi \in \Znn^s \mid \sbi < m\}
\quad\text{and}\quad \psi: \{0,\dots,\nu-1\} \to \Gamma,
\end{equation}
and finally, for $i$ in $\{0,\dots,\mu-1\}$ and $j$ in $\{0,\dots,\nu-1\}$,
associate $M'_i=M_{\phi(i)}$, $N'_j = N_{\psi(j)}$,
$P'_i=P_{\phi(i)}$ and $F'_{i,j} = F_{\phi(i),\psi(j)}$.
Then, we have proved that the solutions to 
Problem~\ref{pbm:multivariate_interpolation} with input parameters 
$s,\ell,n,m_1,\ldots,m_n,b$, $k_1,\ldots,k_s$ and points 
$\{(x_i,y_{i,1},\ldots,y_{i,s})\}_{1\le i\le n}$ are exactly the solutions
to Problem~\ref{pbm:poly_approx} with input parameters 
$\mu,\nu,M'_0,\ldots,M'_{\mu-1}$, $N'_0,\ldots,N'_{\nu-1}$ and polynomials 
$\{(P'_i,F'_{i,0},\ldots,F'_{i,\nu-1})\}_{0\le i<\mu}$. 
This proves the correctness of Algorithm~\ref{algo:reduction}.

\begin{center}
\fbox{\begin{minipage}{14.8cm}
\begin{algo} \textbf{Reducing Problem~\ref{pbm:multivariate_interpolation} to Problem~\ref{pbm:poly_approx}.}
  \label{algo:reduction}

\medskip
  \noindent
  \emph{Input:} 
  $s,\ell,n,m_1,\ldots,m_n$ in $\Zp$,
  $b,k_1,\ldots,k_s$ in $\Z$ 
  and points \\
  \hspace*{1.15cm} $\{(x_i,y_{i,1},\ldots,y_{i,s})\}_{1\le i\le n}$ 
  in $\K^{s+1}$ with the $x_i$ pairwise distinct.

  \medskip
  \noindent
  \emph{Output:} parameters $\mu,\ \nu$, $M'_0,\dots,M'_{\mu-1}$,
 $N'_0,\dots,N'_{\nu-1}$, $\{(P_i,F_{i,0},\ldots,F_{i,\nu-1})\}_{0 \le i < \mu}$\\
 \hspace*{1.56cm}for Problem~\ref{pbm:poly_approx}, such that the solutions
 to this problem are exactly the \\
 \hspace*{1.56cm}solutions to Problem~\ref{pbm:multivariate_interpolation} with parameters the input of this algorithm.

  \begin{enumerate}[{\bf 1.}]
  \item Compute $\Gamma = \{\bj \in \Znn^s \,\mid\, \sbj \leqslant \ell \,\text{ and }\, b - \bj \cdot \bk > 0\}$,
	$\mu = \binom{s+m-1}{s}$, $\nu = |\Gamma|$,
	  and bijections $\phi$ and $\psi$ as in \eqref{eqn:orderings}
  \item Compute $M_\bi = \sum_{1\le r\le n: \; m_r > \sbi} (m_r - \sbi)$ and 
  $N_\bj = b - \bj \cdot \bk$ for $\bj \in \Gamma$
  \item Compute $P_\bi$ and $F_{\bi,\bj}\,$ for $\,\sbi < m,\, \bj \in \Gamma$ as in \eqref{eqn:def_pi_fij}
  \item Return the integers $\mu, \nu, M_{\phi(0)},\dots,M_{\phi(\mu-1)},
 N_{\psi(0)},\ldots,N_{\psi(\nu-1)}$, and the polynomial tuples $\{(P_{\phi(i)},F_{\phi(i),\psi(0)},\ldots,F_{\phi(i),\psi(\nu-1)})\}_{0 \le i < \mu}$\\
  \end{enumerate}
\end{algo}
\end{minipage}}
\end{center}

\begin{prop}
	\label{prop:reduction}
	Algorithm~\ref{algo:reduction} is correct and uses
$\mathcal{O}(\varrho\,\M(M)\log(M))$ operations in $\K$.
\end{prop}
\begin{proof}
The only thing left to do is the complexity analysis;
more precisely, giving an upper bound on the number
of operations in~$\K$ performed in Step \textbf{3}.

First, we need to compute
$P_\bi = \prod_{1\le r\le n:\; m_r > \sbi} (X-x_r)^{m_r-\sbi}$  
for every $\bi$ in $\Znn^s$ such that $\sbi < m$.
This involves only $m$ different polynomials 
$P_{\bi_0},\ldots,P_{\bi_{m-1}}$
where we have chosen any indices $\bi_j$ such that $|\bi_j| = j$.
We note that, defining for $j<m$ the polynomial
$G_j = \prod_{1\le r\le n:\; m_r > j} (X-x_r)$,
we have $P_{\bi_{m-1}} = G_{m-1}$ and for every $j<m-1$,
$P_{\bi_j} = P_{\bi_{j+1}} \cdot G_j$.
The polynomials $G_0,\ldots,G_{m-1}$ have degree at most $n$
and can be computed using $\mathcal{O}(m\M(n)\log(n))$ operations
in $\K$; this is $\mathcal{O}(\varrho\M(M)\log(M))$ since 
$\varrho \ge \binom{s+m-1}{s} \ge m$ and 
$M = \sum_{1\le r\le n} \binom{s+m_r}{s+1} \ge n$. 
Then $P_{\bi_0},\ldots,P_{\bi_{m-1}}$ can be computed 
iteratively using $\mathcal{O}(\sum_{j<m} \M(\deg(P_{\bi_j})))$ operations
in $\K$; using the super-linearity of $\M(\cdot)$, this is 
$\mathcal{O}(\M(M))$ since $\deg(P_{\bi_j}) = M_{\bi_j}$ and 
$\sum_{j<m} M_{\bi_j} \,\le\, M$.

Then, we have to compute (some of) the interpolation polynomials
$R_1,\dots,R_s$. Due to Lemma~\ref{lem:univariate_reformulation},
the only values of $i\in\{1,\ldots,s\}$ for which $R_i$ is needed are those such that
the indeterminate $Y_i$ may actually appear in 
$Q(X,\bY) = \sum_{\bj\in\Gamma} Q_\bj(X) \bY^\bj$.
Now, the latter will not occur unless the $i$th unit $s$-tuple $(0,\ldots,0,1,0,\ldots,0)$
belongs to $\Gamma$. Hence, at most $|\Gamma|$ polynomials $R_i$
must be computed, each at a cost of $\mathcal{O}(\M(n)\log(n))$ operations in $\K$.
Overall, the cost of the interpolation step is thus in $\mathcal{O}(|\Gamma|\M(n)\log(n))
\subseteq \mathcal{O}(\varrho\,\M(M)\log(M))$.

Finally, we compute $F_{\bi,\bj}$ for every
$\bi,\bj$. This is done by fixing $\bi$ and computing all products
$F_{\bi,\bj}$ incrementally, starting from
$R_1,\dots,R_s$. Since we compute modulo $P_{\bi}$,
each product takes $\mathcal{O}(\M(M_\bi))$
operations in $\K$. Summing over all $\bj$ leads to a cost of
$\mathcal{O}(|\Gamma| \M(M_\bi))$ per index $\bi$. Summing over all
$\bi$ and using the super-linearity of $\M$ leads to a total cost of
$\mathcal{O}(|\Gamma| \M(M))$, which is $\mathcal{O}(\varrho \M(M))$.
\end{proof}

The reduction above is deterministic and its cost
is negligible compared to the cost in $\mathcal{O}(\varrho^{\omega-1}\,\M(M)\log(M)^2)$
that follows from Theorem~\ref{thm:poly_approx_complexity} 
with $\rho = \varrho$ and $M' = \sum_{0\le i < \mu}M_i' = M$.
Noting that $M = \sum_{\sbi < m} M_\bi = 
\sum_{1\le r\le n} \binom{s+m_r}{s+1}$,
we conclude that Theorem~\ref{thm:poly_approx_complexity} 
implies Theorem~\ref{thm:complexity}.

\section{Solving Problem~\ref{pbm:poly_approx} through 
structured linear systems}
\label{sec:solutions}

\subsection{Solving structured homogeneous linear systems} 
\label{subsec:structured}

Our two solutions to Problem~\ref{pbm:poly_approx} rely on fast 
algorithms for solving linear systems of the form $Au = 0$ 
with $A$ a structured matrix over $\K$.
In this section, we briefly review useful concepts and results related to {\em
  displacement rank} techniques.  While these techniques can handle
systems with several kinds of structure, we will only need (and
discuss) those related to {\em Toeplitz-like} and {\em Hankel-like}
systems; for a more comprehensive treatment, the reader may
consult~\cite{Pan01}.

Let $M$ be a positive integer and let $\ZZ_M \in \K^{M \times M}$ be
the square matrix with ones on the subdiagonal and zeros elsewhere:
\[\ZZ_M \, = \, \left[
\begin{array}{ccccc}
	0 & 0 & \cdots & 0 & 0 \\
	1 & 0 & \cdots & 0 & 0\\
	0 & 1 & 0 & \cdots & 0\\
	\vdots & \ddots & \ddots & \ddots & \vdots\\
	0 & \cdots & 0 & 1 & 0
\end{array}
\right] \in \K^{M\times M}.\]
Given two integers $M$ and $N$, consider the following operators:
$$\begin{array}{rccc} \Delta_{M,N}:& \K^{M \times N} &\to& \K^{M
    \times N}\\ &A & \mapsto & A-\ZZ_M\,A\,\ZZ^T_N\end{array}$$ 
and
$$\begin{array}{rccc} \Delta'_{M,N}:& \K^{M \times N} &\to& \K^{M
    \times N}\\ &A & \mapsto & A-\ZZ_M\,A\,\ZZ_N,\end{array}$$ which
subtract from $A$ its translate one place along the diagonal,
resp.~along the anti-diagonal. 

Let us discuss $\Delta_{M,N}$ first. If $A$ is a {\em Toeplitz}
matrix, that is, invariant along diagonals, $\Delta_{M,N}(A)$ has rank
at most two. As it turns out, Toeplitz systems can be solved much
faster than general linear systems, in quasi-linear time in $M$.  The
main idea behind algorithms for structured matrices is to extend these
algorithmic properties to those matrices $A$ for which the rank of
$\Delta_{M,N}(A)$ is small, in which case we say that $A$ is {\em
  Toeplitz-like}. Below, this rank will be called the {\em
  displacement rank} of $A$ (with respect to $\Delta_{M,N}$).

Two matrices $(V,\,W)$ in $\K^{M\times \alpha} \times \K^{\alpha\times
  N}$ will be called a {\it generator of length $\alpha$} for $A$
with respect to $\Delta_{M,N}$ if $\Delta_{M,N}(A) = V\,W$. For the
structure we are considering, one can recover $A$ from its generators;
in particular, one can use a generator of length $\alpha$ as a way to
represent $A$ using $\alpha (M+N)$ field elements. One of the main
aspects of structured linear algebra algorithms is to use generators
as a compact data structure throughout the whole process.

Up to now, we only discussed the Toeplitz structure. {\em Hankel-like}
matrices are those which have a small displacement rank with respect
to $\Delta'_{M,N}$, that is, those matrices $A$ for which the rank of
$\Delta'_{M,N}(A)$ is small. As far as solving the system $A u = 0$ is
concerned, this case can easily be reduced to the Toeplitz-like case.
Define $B=A J_N$, where $J_N$ is the reversal matrix of size $N$, all
entries of which are zero, except the anti-diagonal which is set to
one. Then, one easily checks that the displacement rank of $A$ with
respect to $\Delta'_{M,N}$ is the same as the displacement rank of $B$
with respect to $\Delta_{M,N}$, and that if $(V,W)$ is a generator for
$A$ with respect to $\Delta'_{M,N}$, then $(V, W J_N)$ is a generator for
$B$ with respect to $\Delta_{M,N}$. Using the algorithm for
Toeplitz-like matrices gives us a solution $v$ to $B v =0$, from which we
deduce that $u=J_N v$ is a solution to $A u =0$.

In this paper, we will not enter the details of algorithms for solving
such structured systems. The main result we will rely on is the
following proposition, a minor extension of a result by Bostan,
Jeannerod, and Schost~\cite{BoJeSc08}, which features the best known
complexity for this kind of task, to the best of our knowledge. This
algorithm is based on previous work of Bitmead and Anderson~\cite{BA80},
Morf~\cite{M80}, Kaltofen~\cite{K94}, and Pan~\cite{Pan01}, and is
probabilistic (it depends on the choice of some
parameters in the base
field $\K$, and success is ensured provided these parameters avoid a 
hypersurface of the parameter space). 

The proof of the following proposition occupies the rest of this
section. Remark that some aspects of this statement could be improved
(for instance, we could reduce the cost so that it only depends on
$M$, not $\max(M,N)$), but that would be inconsequential for the
applications we make of it.

\begin{prop}\label{prop:struct}
  Given a generator $(V,W)$ of length $\alpha$ for a matrix $A \in
  \K^{M \times N}$, with respect to either $\Delta_{M,N}$ or
  $\Delta'_{M,N}$, one can find a nonzero element in the right nullspace of
  $A$, or determine that none exists, by a probabilistic algorithm
  that uses $\mathcal{O}(\alpha^{\omega-1} \M(P)\log(P)^2)$ operations
  in $\K$, with $P=\max(M,N)$.
  The algorithm chooses $\mathcal{O}(P)$ elements in $\K$; if these
  elements are chosen uniformly at random in a set $S \subseteq \K$ of
  cardinality at least $6P^2$, the probability of success is at least
  $1/2$.
\end{prop}

\paragraph{Square matrices.} 
In all that follows, we consider only the operator $\Delta_{M,N}$,
since we already pointed out that the case of $\Delta'_{M,N}$ can be
reduced to it for no extra cost.

When $M=N$, we use directly~\cite[Theorem~1]{BoJeSc08}, which gives
the running time reported above. That result does not explicitly state
which solution we obtain, as it is written for general non-homogeneous
systems. Here, we want to make sure we obtain a nonzero element in the right nullspace
(if one exists), so slightly more details are needed.
  
The algorithm in that theorem chooses $3M-2$ elements in $\K$, the
first $2M-2$ of which are used to precondition $A$ by giving it
generic rank profile; this is the case when these parameters avoid a
hypersurface of $\K^{2M-2}$ of degree at most $M^2+M$. 

Assume this is the case. Then, 
following~\cite{KaSa91},
the output vector $u$ is obtained in a
parametric form as $u = \lambda(u')$, where $u'$ consists of another set
of $M$ parameters chosen in $\K$ and $\lambda$ is a surjective linear
mapping with image the right nullspace ${\rm ker}(A)$ of~$A$. If ${\rm
  ker}(A)$ is trivial, the algorithm returns the zero vector in any
case, which is correct. Otherwise, the set of vectors $u'$ such that
$\lambda(u')=0$ is contained in a hyperplane of $\K^M$, so it is enough
to choose $u'$ outside of that hyperplane to ensure success.

To conclude we rely on the so-called Zippel-Schwartz lemma~\cite{DeMilloLipton78,Zippel79, Schwartz80},
which can be summarized as follows:
if a nonzero polynomial over $\K$ of total degree at most $d$ is evaluated 
by assigning each of its indeterminates a value chosen
uniformly at random in a subset $S$ of $\K$,
then the probability that the resulting polynomial value be zero is at most $d/|S|$.
Thus, applying that result to the polynomial of degree $d := M^2+M+1 \le 3M^2$
corresponding to the hypersurface and the hyperplane mentioned above,
we see that if we choose all parameters uniformly at
random in a subset $S\subseteq \K$ of cardinality $|S| \ge 6M^2$, the
algorithm succeeds with probability at least $1/2$.


\paragraph{Wide matrices.}
Suppose now that $M < N$, so that the system is underdetermined. We
add $N-M$ zero rows on top of $A$, obtaining an $N \times N$ matrix
$A'$. Applying the algorithm for the square case to $A'$, we will
obtain a right nullspace element $u$ for $A'$ and thus $A$, since these
nullspaces are the same. In order to do so, we need to construct a
generator for $A'$ from the generator $(V,W)$ we have for $A$: one
simply takes $(V',W)$, where $V'$ is the matrix in $\K^{N \times
  \alpha}$ obtained by adding $N-M$ zero rows on top of $V$.

\paragraph{Tall matrices.}
Suppose finally that $M > N$. This time, we build the matrix $A' \in
\K^{M\times M}$ by adjoining $M-N$ zero columns to $A$ on the
left. The generator $(V,W)$ of $A$ can be turned into a generator of
$A'$ by simply adjoining $M-N$ zero columns to $W$ on the left.  We
then solve the system $A' s = 0$, and return the vector $u$ obtained
by discarding the first $M-N$ entries of $s$.

The cost of this algorithm fits into the requested bound; all that
remains to see is that we obtain a nonzero vector in the right nullspace
${\rm ker}(A)$ of $A$ with nonzero probability. Indeed, the nullspaces
of $A$ and $A'$ are now related by the equality ${\rm ker}(A') =
\K^{M-N}\times{\rm ker}(A)$. We mentioned earlier that in the
algorithm for the square case, the solution $s$ to $A's=0$ is obtained
in parametric form, as $s =\lambda(s')$ for $s'\in\K^M$, with $\lambda$ a
surjective mapping $\K^M \to {\rm ker}(A')$. Composing with the
projection $\pi: {\rm ker}(A') \to {\rm ker}(A)$, we obtain a
parametrization of ${\rm ker}(A)$ as $u=(\pi \circ \lambda)(s')$. The
error probability analysis is then the same as in the square case.


\subsection{Solving Problem~\ref{pbm:poly_approx} through a mosaic-Hankel linear system}
\label{subsec:key_equations_version}

In this section, we give our first solution to
Problem~\ref{pbm:poly_approx}, thereby proving
Theorem~\ref{thm:poly_approx_complexity};
this solution is outlined in 
Algorithm~\ref{algo:interpolation_eke} below.
It consists of first deriving and linearizing the 
modular equations of Lemma~\ref{multivariate-zeh-et-al}
below, and then solving the resulting mosaic-Hankel system
using the approach recalled in Section~\ref{subsec:structured}.
Note that, when solving Problem~\ref{pbm:multivariate_interpolation}
using the reduction to Problem~\ref{pbm:poly_approx}
given in Section~\ref{sec:red}, these modular equations 
are a generalization to arbitrary $s$ of the extended key equations presented 
in~\cite{RotRuc00,ZeGeAu11,Zeh13} for $s=1$.

We consider input polynomials $\{(P_i,\bF_{i})\}_{0 \le i < \mu}$ with,
for all $i$, $P_i$ monic of degree $M'_i$ and $\bF_i$ a vector of
$\nu$ polynomials $(F_{i,0},\dots,F_{i,\nu-1})$, all of degree less
than $M'_i$. Given degree bounds $N'_0,\dots,N'_{\nu-1}$, we look for
polynomials $Q_0,\ldots,Q_{\nu-1}$ in $\K[X]$ such that the
following holds:
\begin{enumerate}[$(a)$]
\item\label{cond:ip} the $Q_j$ are not all zero,
\item\label{cond:iip} for $0 \le j<\nu$, $\deg(Q_j) < N'_j,$
\item\label{cond:iiip} for $0 \le i<\mu$, $\sum_{0 \le j <\nu} F_{i,j}Q_j   =  0 \bmod{P_i}$.
\end{enumerate}

Our goal here is to linearize Problem~\ref{pbm:poly_approx} into a homogeneous linear system over $\K$
involving $M'$ linear equations with $N'$ unknowns,
where $M' = M'_0 + \cdots + M'_{\mu-1}$ and 
$N' = N'_0 + \cdots + N'_{\nu-1}$.
Without loss of generality, we will assume that 
\begin{equation}
	\label{eqn:quasi-square_system}
	N' \le   M'+1. 
\end{equation}
Indeed, if $N' \ge M'+1$, the instance of
Problem~\ref{pbm:poly_approx} we are considering has more unknowns
than equations. We may set the last $N'-(M'+1)$ unknowns to zero,
while keeping the system underdetermined. This simply amounts to
replacing the degree bounds $N'_0,\dots,N'_{\nu-1}$ by
$N'_0,\dots,N'_{\nu'-2},N''_{\nu'-1}$, for $\nu' \le \nu$ and
$N''_{\nu'-1} \le N'_{\nu'-1}$ such that
$N'_0+\cdots+N'_{\nu'-2}+N''_{\nu'-1}=M'+1$. In particular, $\nu$
may only decrease through this process.

In what follows, we will work with the reversals of the input and output polynomials
of Problem~\ref{pbm:poly_approx}, defined by 
\[
{\overline {P_i}}= X^{M'_i}P_i(X^{-1}), \qquad
\overline{F_{i,j}} = X^{M'_i-1} F_{i,j}(X^{-1}), \qquad
\overline{Q_j} = X^{N'_j-1}Q_j(X^{-1}).
\]
Let also $\beta = \max_{h < \nu} N_h'$ and, 
for $0 \le i <\mu$ and $0 \le j <\nu$,
\[
\delta_i = M'_i + \beta - 1 \quad\text{and}\quad
\gamma_j = \beta - N'_j.
\]
In particular, $\delta_i$ and $\gamma_j$ are nonnegative integers
and, recalling that $P_i$ is monic, we can define further
the polynomials $S_{i,j}$ in $\K[X]$ as 
\[
 S_{i,j} = \frac{X^{\gamma_j} \overline{F_{i,j}}}{ \overline   {P_i}} \bmod X^{\delta_i}
\]
for $0 \le i < \mu$ and $0 \le j < \nu$.
By using these polynomials, we can now reformulate the approximation condition of Problem~\ref{pbm:poly_approx}
in terms of a set of extended key equations:

\begin{lem} \label{multivariate-zeh-et-al}
Let $Q_0,\ldots,Q_{\nu-1}$ be polynomials in $\K[X]$ that
satisfy condition~\eqref{cond:iip} in Problem~\ref{pbm:poly_approx}. 
They satisfy condition~\eqref{cond:iiip} in Problem~\ref{pbm:poly_approx} 
if and only if for all $i$ in $\{0,\dots,\mu-1\}$, there exists a polynomial $T_i$ in $\K[X]$ such
that
\begin{equation} \label{eq:modular-equations}
\sum_{0 \le j  < \nu} S_{i,j} \overline{Q_j} = T_i \bmod X^{\delta_i}
\qquad \text{and}\qquad 
\deg(T_i) < \beta-1.
\end{equation}
\end{lem}
\begin{proof}
Condition~\eqref{cond:iiip} holds if and only if for all $i$ in
$\{0,\dots,\mu-1\}$, there exists a polynomial $B_i$ in $\K[X]$ such
that
\begin{equation} \label{eq:divisor-Bi}
\sum_{0 \le j <\nu} F_{i,j} Q_j = B_i P_i.
\end{equation}
For all $i,j$, the summand $F_{i,j}Q_j$ has degree less than $M'_i+N'_j-1$, 
so the left-hand term above has degree less than $\delta_i$.
Since $P_i$ has degree $M'_i$, this implies that whenever a polynomial
$B_i$ as above exists, we must have $\deg(B_i) < \delta_i-M'_i = \beta - 1$.
Now, by substituting $1/X$ for $X$ and multiplying by $X^{\delta_i-1}$ we
can rewrite the identity in~(\ref{eq:divisor-Bi}) as
\begin{equation} \label{eq:reverse-of-divisor-Bi}
  \sum_{0 \le j <\nu} \overline{F_{i,j}}\, \overline{Q_j} X^{\gamma_j}
= T_i \overline{P_i},
\end{equation}
where $T_i$ is the polynomial of degree less than $\beta - 1$
given by $T_i = X^{\beta-2}B_i(X^{-1})$.  Since the degrees of
both sides of~(\ref{eq:reverse-of-divisor-Bi}) are less than
$\delta_i$, one can consider the above identity modulo
$X^{\delta_\bi}$ without loss of generality, and since
$\overline{P_i}(0)=1$ one can further divide by $\overline{P_i}$
modulo $X^{\delta_i}$.  This shows
that~\eqref{eq:reverse-of-divisor-Bi} is equivalent to the identity
in~\eqref{eq:modular-equations}, and the proof is complete.
\end{proof}

Following~\cite{RotRuc00,ZeGeAu11}, we are going to rewrite the latter
conditions as a linear system in the coefficients of the polynomials
$Q_0,\ldots,Q_{\nu-1}$, eliminating the unknowns $T_i$ from the outset. 
Let us first define the \emph{coefficient vector} of 
a solution $(Q_0,\ldots,Q_{\nu-1})$ to Problem~\ref{pbm:poly_approx}
as the vector in $\K^{N'}\!$ obtained by concatenating, for $0 \le j < \nu$, 
the vectors $\big[Q_j^{(0)},Q_j^{(1)},\ldots,Q_j^{(N_j'-1)}\big]^T$ of the coefficients of $Q_j$.
Furthermore, 
denoting by $S_{i,j}^{(0)},S_{i,j}^{(1)},\ldots, S_{i,j}^{(\delta_i-1)}$
the $\delta_i \ge 1$ coefficients of the polynomial $S_{i,j}$, 
we set up the block matrix \[ A = \big[A_{i,j}\big]_{0 \le i<\mu , \, 0 \le j < \nu} \in \K^{M' \times N'}, \]
whose block $(i,j)$ is
the Hankel matrix 
\[
A_{i,j} = \big[S_{i,j}^{(u+v + \gamma_j)}\big]_{0 \le u<M'_i, \, 0\le v<N'_j} \;\in\K^{M'_i\times N'_j}.
\]

%

\begin{lem}
  \label{lem:eke_nonzero_nullspace}
  A nonzero vector of $\K^{N'}\!$ is in the right nullspace of $A$ if and only
  if it is the coefficient vector of a solution $(Q_0,\ldots,Q_{\nu-1})$ to
  Problem~\ref{pbm:poly_approx}.
\end{lem}
\begin{proof}
  It is sufficient to consider a polynomial tuple $(Q_0,\ldots,Q_{\nu-1})$ that
  satisfies~\eqref{cond:iip}.  Then, looking at the high-degree terms
  in the identities in~\eqref{eq:modular-equations}, we see that
  condition~\eqref{cond:iiip} is equivalent to the following homogeneous system of
  linear equations over $\K$: for all $i$ in $\{0,\dots,\mu-1\}$
  and all $\delta$ in $\{\delta_i-M'_i,\dots,\delta_i-1\}$, 
\begin{equation*}
  \label{eqn:extended_key_equations}
  \sum_{0 \le j < \nu} \; \sum_{0\le r < N'_j} S_{i,j}^{(\delta-r)} Q_j^{(N'_j-1-r)} = 0.
\end{equation*}
  The matrix obtained by considering all these equations is precisely the matrix $A$.
\end{proof}

We will use the approach recalled in Section~\ref{subsec:structured} to find a
nonzero nullspace element for~$A$, with respect to the displacement operator
$\Delta'_{M',N'}$. Not only do we need to prove that the displacement
rank of $A$ with respect to $\Delta'_{M',N'}$ is bounded by a value $\alpha$ not too large, 
but we also have to efficiently compute a generator of length $\alpha$ for $A$, that is, 
a matrix pair $(V,W)$ in $\K^{M'\times\alpha} \times \K^{\alpha\times N'}$
such that $A - \ZZ_{M'} A \,\ZZ_{N'} = VW$. We will see that here,
computing such a generator boils down to computing the coefficients of
the polynomials $S_{i,j}$.
The cost incurred by computing this generator is summarized in the
following lemma; combined with Proposition~\ref{prop:struct} and 
Lemma~\ref{lem:eke_nonzero_nullspace}, this proves
Theorem~\ref{thm:poly_approx_complexity}.

\begin{lem}
  \label{lem:generators_key_equations}
  The displacement rank of $A$ with respect to $\Delta'_{M',N'}$ is at
  most $\mu+\nu$. Furthermore, a corresponding generator of length $\mu+\nu$
  for $A$ can be computed
  using $\mathcal{O}\left((\mu+\nu) \M(M')\right)$ operations in $\K$.
\end{lem}
\begin{proof}
  We are going to exhibit two matrices $V\in \K^{M'\times (\mu+\nu)}$
  and $W \in \K^{(\mu+\nu) \times N'}$ such that $A - \ZZ_{M'} A
  \,\ZZ_{N'} = VW$.  Because of the structure of $A$, at most $\mu$
  rows and $\nu$ columns of the matrix $A - \ZZ_{M'} A \,\ZZ_{N'} = A
  - (A \text{ shifted left and down by one unit})$ are nonzero. More
  precisely, only the first row and the last column of each $M'_i
  \times N'_j$ block of this matrix can be nonzero. Indexing the
  rows, resp.\ columns, of $A - \ZZ_{M'} A \,\ZZ_{N'}$ from $0$ to
  $M'-1$, resp.\ from $0$ to $N'-1$, only the $\mu$ rows with indices
  of the form $r_i = M'_0 + \cdots + M'_{i-1}$ for
  $i=0,\dots,\mu-1$ can be nonzero, and only the $\nu$ columns with
  indices of the form $c_j = N'_0 + \cdots + N'_j - 1$ for
  $j=0,\dots,\nu-1$ can be nonzero.

  For any integers $0\le i< K$, define 
  $\OO_{i,K} = \left[0 \; \cdots \; 0 \; 1 \; 0 \; \cdots \; 0\right]^T\in\K^{K}$ 
  with $1$ at position~$i$, and
  \[\OO^{(V)} = \left[{\OO_{r_i,M'}}\right]_{0 \le i < \mu} \in\K^{M' \times\mu},\quad 
  \OO^{(W)} = \left[{\OO_{c_j,N'}}\right]^T_{0 \le j < \nu} \in \K^{\nu \times N'}. \]
  For given $i$ in $\{0,\dots,\mu-1\}$ and $j$ in $\{0,\dots,\nu-1\}$,
  we will consider $v_{i,j} = [v_{i,j}^{(r)}]_{0\le r<M'_i}$ in $\K^{M'_i \times 1}$ 
  and $w_{i,j} = [w_{i,j}^{(r)}]_{0 \le r<N'_j}$ in $\K^{1\times N'_j}$, 
  which are respectively the last column
  and the first row of the block $(i,j)$ in $A - \ZZ_{M'} A\, \ZZ_{N'}$,
  up to a minor point: the first entry of $v_{i,j}$ is set to zero. 
  The coefficients $v_{i,j}^{(r)}$ and $w_{i,j}^{(r)}$ 
  can then be expressed in terms of the entries $A_{i,j}^{(u,v)} = S_{i,j}^{(u+v + \gamma_j)}$ of the Hankel matrix 
  $A_{i,j} = [A_{i,j}^{(u,v)}]_{0 \le u<M'_i, \, 0\le v<N'_j}$ as follows:
  \begin{align}
  v_{i,j}^{(r)} &= \left\{
  \begin{array}{ll}
    0 & \text{if $r=0$,}\\[2mm]
    A_{i,j}^{(r,N'_j-1)} - A_{(i,j+1)}^{(r-1,0)} & \text{if $1 \le r < M_i'$,}
  \end{array}\right. \label{eqn:def_vij}\\[4mm] 
  w_{i,j}^{(r)} &= \left\{
  \begin{array}{ll}
    A_{i,j}^{(0,r)}-A_{i-1,j}^{(M'_{i-1}-1,r+1)} & \text{if } r < N'_j - 1,\\[2mm]
    A_{i,j}^{(0,N'_j-1)}-A_{i-1,j+1}^{(M'_{i-1}-1,0)} & \text{if } r = N'_j - 1. \label{eqn:def_wij}
  \end{array}\right.
  \end{align}
  Note that here, we use the convention that an indexed object is zero 
  when the index is out of the allowed bounds for this object.

  Then, we define $V_j$ and $W_i$ as
  \[V_j = \left[\begin{matrix} 
    v_{0,j} \\ \vdots \\ v_{\mu-1,j} 
    \end{matrix} \right] \in \K^{M'\times 1} \quad\text{and}\quad
    W_i = \left[ w_{i,0}~\cdots~w_{i,\nu-1}\right] \in\K^{1\times N'},
  \]
  and
  \[V' = \left[V_0 \,\cdots\, V_{\nu-1}\right] \in\K^{M'\times \nu}\quad\text{and}\quad
  W' = \left[\begin{matrix} W_0 \\ \vdots \\ W_{\mu-1} \end{matrix}\right] \in \K^{\mu \times N'}.
  \]
  Now, one can easily verify that the matrices
  \begin{equation}
  \label{eqn:def_generators_eke}
  V = \left[ V' ~~ \OO^{(V)} \right] \; \in \K^{M'\times (\mu+\nu)} 
  \qquad\text{and}\qquad
  W = \left[\begin{matrix} \OO^{(W)} \\ W' \end{matrix} \right] \; \in \K^{(\mu+\nu)\times N'}
  \end{equation}
  are generators for $A$, that is, $A - \ZZ_M A \,\ZZ_N = VW$.
  
  We notice that all we need to compute the generators $V$ and $W$ are
  the last $M'_i+N'_j-1$ coefficients of 
  $S_{i,j}(X) = S_{i,j}^{(0)} + S_{i,j}^{(1)} X + \cdots + S_{i,j}^{(\delta_i-1)}X^{\delta_i-1}$ 
  for every $i$ in
  $\{0,\dots,\mu-1\}$ and $j$ in $\{0,\dots,\nu-1\}$. Now, recall that
  $$S_{i,j} = \frac{X^{\gamma_j} \overline{F_{i,j}}}{ \overline {P_i}} \bmod
  X^{\delta_i} = \frac{X^{\delta_i-(M'_i+N'_j-1)} \overline{F_{i,j}}}{ \overline
    {P_i}} \bmod X^{\delta_i}.$$ Thus, the first
  $\delta_i-(M'_i+N'_j-1)$ coefficients of $S_{i,j}$ are zero, and the
  last $M'_i+N'_j-1$ coefficients of $S_{i,j}$ are the coefficients of
  \begin{equation}
  \label{eqn:def_Sstar}
  S^\star_{i,j} = \frac{\overline{F_{i,j}}}{ \overline {P_i}} \bmod
  X^{M'_i+N'_j-1},
  \end{equation}
  which can be computed in
  $\mathcal{O}(\M(M'_i+N'_j))$ operations in $\K$ by fast power series
  division. By expanding products, we see that
  $\M(M'_i+N'_j)=\mathcal{O}(\M(M'_i)+\M(N'_j))$. Summing the costs,
  we obtain an upper bound of the form
  $$\mathcal{O} \left (\sum_{0 \le i < \mu} \sum_{0 \le j < \nu}
  \M(M'_i)+\M(N'_j) \right ).$$ Using the super-linearity of $\M$,
  this is in $\mathcal{O} (\nu \M(M') + \mu \M(N'))$. Since we
  assumed in~\eqref{eqn:quasi-square_system} that $N' \le M'+1$, this is $\mathcal{O} ((\mu+\nu) \M(M') )$.
\end{proof}

\begin{center}
\fbox{\begin{minipage}{15cm}
\begin{algo} \textbf{Solving Problem~\ref{pbm:poly_approx} via a mosaic-Hankel linear system.}
  \label{algo:interpolation_eke}

\medskip
\noindent
\emph{Input:}
 positive integers $\mu,\ \nu$, $M'_0,\dots,M'_{\mu-1}$,
 $N'_0,\dots,N'_{\nu-1}$ and polynomial tuples \\
 \hspace*{1.15cm} $\{(P_i,F_{i,0},\ldots,F_{i,\nu-1})\}_{0 \le i < \mu}$ in $\K[X]^{\nu+1}$
 such that for all $i$, $P_i$ is monic of \\
 \hspace*{1.15cm} degree $M'_i$ and $\deg(F_{i,j})< M'_i$ for all $j$.

\medskip
\noindent
\emph{Output:}
polynomials $Q_0,\ldots,Q_{\nu-1}$ in $\K[X]$ such that \eqref{cond:a}, \eqref{cond:b}, \eqref{cond:c}.

  \begin{enumerate}[{\bf 1.}]
  \item For $i<\mu$, $j<\nu$, compute the coefficients $S_{i,j}^{(\gamma_j + r)}$ for $r< M'_i + N'_j-1$,
	  that is, the coefficients of the polynomials $S^\star_{i,j}$ defined in~\eqref{eqn:def_Sstar}
	\item For $i<\mu$ and $j<\nu$, compute the vectors $v_{i,j}$ and $w_{i,j}$ 
	as defined in \eqref{eqn:def_vij} and \eqref{eqn:def_wij}
	\item For $i<\mu$, compute $r_i = M'_0 + \cdots + M'_{i-1}\,$;
	for $j<\nu$, compute $c_j = N'_0 + \cdots + N'_j - 1$
	\item Deduce the generators $V$ and $W$ as defined in~\eqref{eqn:def_generators_eke}
	from $r_i,c_j,v_{i,j},w_{i,j}$
	\item Use the algorithm of Proposition~\ref{prop:struct} with input $V$ and $W$; if there is no solution then exit with no solution, otherwise find the coefficients of $Q_0,\ldots,Q_{\nu-1}$
	\item Return $Q_0,\ldots,Q_{\nu-1}$
  \end{enumerate}

\end{algo}
\end{minipage}}
\end{center}


\subsection{A direct solution to Problem~\ref{pbm:poly_approx}}
\label{subsec:direct}

In this section, we propose an alternative solution to
Problem~\ref{pbm:poly_approx}
which leads to the same asymptotic running time as in the previous section
but avoids the extended key equations of Lemma~\ref{multivariate-zeh-et-al};
it is outlined in 
Algorithm~\ref{algo:interpolation_direct} below.
As above, our input consists of the polynomials $(P_i,F_{i,0},\ldots,F_{i,\nu-1})_{0 \le i < \mu}$
and we look for polynomials $Q_0,\ldots,Q_{\nu-1}$ in $\K[X]$ such
that for $0 \le i<\mu$, $\sum_{0 \le j <\nu} F_{i,j} Q_j = 0
\bmod{P_i}$, with the $Q_j$ not all zero and for $j<\nu$, $\deg Q_j \,<\, N'_j$.

In addition, for $r \ge 0$, we denote by $F_{i,j}^{(r)}$ and $P_i^{(r)}$ the
coefficients of degree $r$ of $F_{i,j}$ and $P_i$, respectively, and we define
$\mathcal{C}_i$ as the ${M'_i \times M'_i}$ companion matrix of $P_i$;
if $B$ is a polynomial of degree less than $M'_i$ with coefficient
vector $v \in \K^{M'_i}$, then the product $\mathcal{C}_i 
v \,\in \K^{M'_i}$ is the coefficient vector of the polynomial $X
B \bmod{P_i}$.  Explicitly, we have
\[\mathcal{C}_i \;=\; \left[\begin{matrix}
    0 & 0 & \cdots & 0 & -P_i^{(0)} \\
    1 & 0 & \cdots & 0 & -P_i^{(1)} \\
    0 & 1 & \cdots & 0 & -P_i^{(2)} \\
    \vdots & \vdots & \ddots & \vdots & \vdots \\
    0 & 0 & \cdots & 1 & -P_i^{(M'_i-1)} 
  \end{matrix}\right] \in \K^{M_i'\times M_i'}.\]

We are going to see that solving Problem~\ref{pbm:poly_approx} is
equivalent to finding a nonzero solution to a homogeneous linear
system whose matrix is $A'=(A'_{i,j}) \in \K^{M'\times N'}$, where for
every $i<\mu$ and $j<\nu$, $A'_{i,j} \in \K^{M'_i \times N'_j}$ is a
matrix which depends on the coefficients of $F_{i,j}$ and $P_i$.
Without loss of generality, we make the same assumption as in the
previous section, that is, $N' \le M'+1$ holds.

For $i,j$ as above and for $h \in \Znn$, let $\alpha_{i,j}^{(h)} \in
\K^{M'_i}$ be the coefficient vector of the polynomial $X^h F_{i,j}
\bmod{P_i}$, so that these vectors are given by
$$\alpha_{i,j}^{(0)} = 
\begin{bmatrix}
  F_{i,j}^{(0)}\\
  \vdots \\
  F_{i,j}^{(M'_i-1)}
\end{bmatrix}
\quad\text{and}\quad \alpha_{i,j}^{(h+1)} = \mathcal{C}_i\,
\alpha_{i,j}^{(h)}.$$ Let then $A'=(A'_{i,j}) \in \K^{M'\times N'}$,
where for every $i<\mu$ and $j<\nu$, the block $A'_{i,j} \in \K^{M'_i
  \times N'_j}$ is defined by
$$A'_{i,j} =  \begin{bmatrix} \alpha_{i,j}^{(0)} & \cdots &\alpha_{i,j}^{(N'_j-1)} 
\end{bmatrix}.$$

\begin{lem}
  \label{lem:direct_nonzero_nullspace}
  A nonzero vector of $\K^{N'}\!$ is in the right nullspace of $A'$ if and only
  if it is the coefficient vector of a solution $(Q_0,\ldots,Q_{\nu-1})$ to
  Problem~\ref{pbm:poly_approx}.
\end{lem}
\begin{proof}
By definition $A'_{i,j}$ is the $M_i'\times N_j'$ matrix 
of the mapping $Q \mapsto F_{i,j} Q \bmod P_i$, for $Q$ in $\K[X]$ 
of degree less than $N'_j$.
Thus, if $(Q_0,\dots,Q_{\nu-1})$ is a $\nu$-tuple of polynomials that
satisfies the degree constraint~\eqref{cond:iip} in Problem~\ref{pbm:poly_approx},
applying $A'$ to the coefficient vector of this tuple outputs the
coefficients of the remainders $\sum_{0 \le j <\nu} F_{i,j}Q_j 
\bmod{P_i}$, for $i=0,\dots,\mu-1$. 
The claimed equivalence then follows immediately.
\end{proof}

The following lemma shows that $A'$ possesses a Toeplitz-like
structure, with displacement rank at most $\mu+\nu$. Together with
Proposition~\ref{prop:struct} and Lemma~\ref{lem:direct_nonzero_nullspace}, 
this gives our second proof of Theorem~\ref{thm:poly_approx_complexity}.

\begin{lem}
  \label{prop:generators_poly_approx}
  The displacement rank of $A'$ with respect to $\Delta_{M',N'}$ is at
  most $\mu+\nu$. Furthermore, a corresponding generator of length $\mu+\nu$ for $A'$
  can be computed using $\mathcal{O}((\mu+\nu) \M(M'))$ operations in $\K$.
\end{lem}
\begin{proof}
  We begin by constructing some matrices $Y\in \K^{M'\times (\mu+\nu)}$ and
  $Z \in \K^{(\mu+\nu)\times N'}$ such that $\Delta_{M',N'}(A')$ is equal to the product $YZ$.  Define first the matrix
  \[\mathcal{C} = \left[\begin{matrix}
      \mathcal{C}_0 & 0 & \cdots & 0 \\ 0 & \mathcal{C}_1 & \cdots & 0
      \\ \vdots & \vdots & \ddots & \vdots \\ 0 & 0 & \cdots &
      \mathcal{C}_{\mu-1}
    \end{matrix}\right] \in\K^{M'\times M'}.\]
  Up to $\mu$ columns, $\mathcal{C}$ coincides with $\mathcal{Z}_{M'}$;
  we make this explicit as follows. For $i$ in $\{0,\dots,\mu-~1\}$, 
  define
  \begin{equation}
  \label{eqn:def_viVi}
  v_i = \left[\begin{matrix} P_i^{(0)} \\ \vdots \\ P_i^{(M'_i-1)} \end{matrix}\right] \in\K^{M'_i}, \qquad
  V_i = \left[\begin{matrix} 0\\ \vdots\\ 0 \\ v_i \\ 1 \\ 0 \\ \vdots
      \\ 0 \end{matrix}\right] \in\K^{M'}, \qquad 
  W_i = \left[\begin{matrix} 0 \\ \vdots \\ 0\\ 1\\ 0 \\ 0
      \\ \vdots\\ 0\end{matrix}\right] \in \K^{M'},
  \end{equation}
  where the last entry of $v_i$ in $V_i$ and the coefficient $1$ in $W_i$
  have the same index, namely $M_0' + \cdots + M_i'-1$.
  (Hence the last vector $V_{\mu-1}$ only contains $v_{\mu-1}$, without a $1$ after it.)
  Then, defining $V = [V_0 \,\cdots\, V_{\mu-1}] \in\K^{M'\times \mu}$ and $W =
  [W_0 \,\cdots\, W_{\mu-1}] \in \K^{M'\times \mu}$, 
  we obtain
  \[ \mathcal{C} = \mathcal{Z}_{M'} - V_0 W_0^T - \cdots - V_{\mu-1} W_{\mu-1}^T \;=\; \mathcal{Z}_{M'} - VW^T.\]
  As before, we use the convention that an indexed object is zero 
  when the index is out of the allowed bounds for this object.
  For $j$ in $\{0,\dots,\nu-1\}$, let us further define
  \begin{equation}
  \label{eqn:def_vpiWpi}
  V'_j = \left[\begin{matrix} \alpha_{0,j}^{(0)} \\ \vdots \\ \alpha_{\mu-1,j}^{(0)} \end{matrix}\right] \;-\;
  \left[\begin{matrix} \alpha_{0,j-1}^{(N'_{j-1})} \\ \vdots
      \\ \alpha_{\mu-1,j-1}^{(N'_{j-1})} \end{matrix}\right]\in\K^{M'}
  \quad\text{and}\quad W'_j = \left[\begin{matrix}0
      \\ \vdots\\ 0\\ 1\\ 0 \\ \vdots\\ 0\end{matrix}\right] \in
  \K^{N'},
  \end{equation}
  with the coefficient $1$ in $W'_j$ at index $N_0'+\cdots+N_{j-1}'$, 
  and the compound matrices
  $$V' = [V'_0 \,\cdots\, V'_{\nu-1}] \,\in\K^{M'\times \nu} \quad\text{and}\quad W' = [W'_0 \,\cdots\, W'_{\nu-1}] \,\in\K^{N'\times \nu}.$$
  Then, we claim that the matrices 
  \begin{equation}
  \label{eqn:generatorsYZ}
  Y = \left[\, -V ~~~ V' \,\right] \in\K^{M' \times (\mu+\nu)}
  \quad\text{and}\quad 
  Z =  \begin{bmatrix} W^T \! A' \mathcal{Z}_{N'}^T \\ {W'}^T \end{bmatrix} \in\K^{ (\mu+\nu)\times N'}
  \end{equation}
  are generators for $A'$ for the Toeplitz-like displacement structure,
  {\it i.e.}, that \[A' - \mathcal{Z}_{M'} \,A'\, \mathcal{Z}_{N'}^T = Y
Z.\] By construction, we have $\mathcal{C}\, A' =
(B_{i,j})_{i<\mu,j<\nu} \,\in\K^{M'\times N'}$, with $B_{i,j}$ given
by
\[B_{i,j} \;=\; \mathcal{C}_i A'_{i,j}  \;=\;  \left[\alpha_{i,j}^{(1)} 
~ \cdots ~ \alpha_{i,j}^{(N'_j-1)} ~ \alpha_{i,j}^{(N'_j)}\right] 
\,\in\K^{M'_i \times N'_j}.\]
As a consequence, 
$A' - \mathcal{C}\, A' \, \mathcal{Z}_{N'}^T = V' W'^T,$ so finally we get, as claimed,
  \begin{align*}
    A' - \mathcal{Z}_{M'} \,A'\, \mathcal{Z}_{N'}^T \quad & = \;\; A' - (\mathcal{C} + VW^T)A' \mathcal{Z}_{N'}^T \\
    & = \;\; A' - \mathcal{C}\, A'\, \mathcal{Z}_{N'}^T \,-\, VW^T \! A' \mathcal{Z}_{N'}^T \\
    & = \;\; V' W'^T \,-\, VW^T \! A\, \mathcal{Z}_{N'}^T \\
    & = \;\; Y Z.
  \end{align*}

  To compute $Y$ and $Z$, the only non-trivial steps are those giving
  $V'$ and $W^T \! A'$. For the former, we have to compute the
  coefficients of $X^{N'_j} F_{i,j} \bmod P_i$ for every $i<\mu$
  and $j<\nu-1$.  For fixed $i$ and $j$, this can be done using fast
  Euclidean division in $\mathcal{O}(\M(M'_i+N'_j))$ operations in
  $\K$, which is $\mathcal{O}(\M(M'_i)+\M(N'_j))$.  Summing over the
  indices $i<\mu$ and $j<\nu-1$, this gives a total cost of
  $\mathcal{O}(\nu \M(M') + \mu \M(N'))$ operations. 
  This is $\mathcal{O}((\mu+\nu) \M(M'))$, since  by assumption $N' \le M'+1$.

  Finally, we show that $W^T \! A'$ can be computed using
  $\mathcal{O}((\mu+\nu) \M(M'))$ operations as well.  Computing this
  matrix amounts to computing the rows of $A'$ of indices 
  $M_0' + \cdots + M_i' - 1$,
  for $i<\mu$. By construction of $A'$, this
  means that we want to compute the coefficients of degree $M'_i-1$ of
  $X^h F_{i,j} \bmod{P_i}$ for $h=0,\dots,N'_j-1$ and for all $i,j$.
  Unfortunately, the naive approach leads to a cost proportional to
  $M'N'$ operations, which is not acceptable. However, for $i$ and $j$
  fixed, Lemma~\ref{lem:fast_generator_computation} below shows how to
  do this computation using only $\mathcal{O}(\M(M'_i)+\M(N'_j))$
  operations, which leads to the announced cost by summing over $i$
  and~$j$.
\end{proof}

\begin{lem}
  \label{lem:fast_generator_computation}
  Let $P \in \K[X]$ be monic of degree $m$, let $F \in \K[X]$ be of
  degree less than $m$, and for $i \ge 0$ let $c_i$ denote the
  coefficient of degree $m-1$ of $X^i F \bmod{P}$.  For $n\ge 1$ we
  can compute $c_0,\ldots,c_{n-1}$ using $\mathcal{O}(\M(m)+\M(n))$
  operations in $\K$.
\end{lem}
\begin{proof}
  Writing $F = \sum_{0 \le j < m} f_j X^j$ we have $X^iF \bmod{P} =
  \sum_{0 \le j < m} f_j \big(X^{i+j} \bmod P\big) $.  Hence $c_i =
  \sum_{0 \le j < m} f_j b_{i+j}$, with $b_i$ denoting the coefficient 
  of degree $m-1$ of $X^i \bmod P$.  Since $b_0 = \cdots = b_{m-2} =
  0$ and $b_{m-1}=1$, we can deduce $c_0,\ldots,c_{n-1}$ from
  $b_{m-1},b_m,\ldots,b_{m+n-2}$ in time $\mathcal{O}(\M(n))$ by
  multiplication by the lower triangular Toeplitz matrix
  $[f_{m+j-i-1}]_{i,j}$ of order $n-1$.  

  Thus, we are left with the question of computing the $n-1$ coefficients
  $b_m,\ldots,b_{m+n-2}$. Writing $P$ as $P = X^m + \sum_{0 \le j <
    m}p_j X^j$ and using the fact that $X^i P \bmod{P} = 0$ for
  all $i\ge 0$, we see that the $b_i$ are generated by a linear
  recurrence of order $m$ with constant coefficients:
  \[
  b_{i+m} + \sum_{0 \le j < m} p_j b_{i+j} = 0 \quad \mbox{for all $i \ge 0$.}
  \]
  Consequently, $b_m,\ldots,b_{m+n-2}$ can be deduced from
  $b_0,\ldots,b_{m-1}$ in time $\mathcal{O}(\frac{n}{m} \M(m))$,
  which is $\mathcal{O}(\M(m)+\M(n))$, by $\lceil \frac{n-1}{m} \rceil$ calls to
  Shoup's algorithm for extending a linearly recurrent
  sequence~\cite[Theorem~3.1]{Shoup91}. 
\end{proof}

\begin{center}
\fbox{\begin{minipage}{15cm}
\begin{algo} \textbf{Solving Problem~\ref{pbm:poly_approx} via a Toeplitz-like linear system.}
  \label{algo:interpolation_direct}

\medskip
\noindent
\emph{Input:}
 positive integers $\mu,\ \nu$, $M'_0,\dots,M'_{\mu-1}$,
 $N'_0,\dots,N'_{\nu-1}$ and polynomial tuples \\
 \hspace*{1.15cm} $\{(P_i,F_{i,0},\ldots,F_{i,\nu-1})\}_{0 \le i < \mu}$ in $\K[X]^{\nu+1}$
 such that for all $i$, $P_i$ is monic of \\
 \hspace*{1.15cm} degree $M'_i$ and $\deg(F_{i,j})< M'_i$ for all $j$.

\medskip
\noindent
\emph{Output:}
polynomials $Q_0,\ldots,Q_{\nu-1}$ in $\K[X]$ such that \eqref{cond:a}, \eqref{cond:b}, \eqref{cond:c}.

  \begin{enumerate}[{\bf 1.}]
  	\item Compute $v_i$ and $V_i$ for $i<\mu$, as defined in~\eqref{eqn:def_viVi}; compute $V = [V_0 \; \cdots \; V_{\mu-1}]$
	\item Compute $W'_j$ for $j<\nu$, as defined in~\eqref{eqn:def_vpiWpi}; compute $W' = [W'_0 \; \cdots \; W'_{\nu-1}]$
	\item Compute $\alpha_{i,j}^{(N'_j)}$, that is, the coefficients of $X^{N'_j} F_{i,j} \bmod P_i$, for $i<\mu,\, j<\nu-1$ (e.g. using fast Euclidean division)
	\item Compute $V'_j$ for $j<\mu$, as defined  in~\eqref{eqn:def_vpiWpi}; compute $V' = [V'_0 \;\cdots\; V'_{\nu-1}]$
  	\item Compute the row of index $M'_0+\cdots+M'_i -1$ of $A'$, for $i<\mu$, that is, the coefficient of degree $M'_i-1$ of $X^h F_{i,j} \bmod{P_i}$, for $h<N'_j,\,j<\nu$ (see the proof of Lemma~\ref{lem:fast_generator_computation} for fast computation).
	\item Compute $W^T \! A'$ whose row of index $i$ is the row of index $M'_0+\cdots+M'_i -1$ of $A'$
	\item Compute the generators $Y$ and $Z$ as defined in~\eqref{eqn:generatorsYZ}
	\item Use the algorithm of Proposition~\ref{prop:struct} with input $Y$ and $Z$; if there is no solution then exit with no solution, otherwise find the coefficients of $Q_0,\ldots,Q_{\nu-1}$
	\item Return $Q_0,\ldots,Q_{\nu-1}$
  \end{enumerate}

\end{algo}
\end{minipage}}
\end{center}

\section{Applications to the decoding of Reed-Solomon codes}
\label{sec:rs_codes}

To conclude, we discuss Theorem~\ref{thm:complexity} in specific
contexts related to the decoding of Reed-Solomon codes; in this
section we always have $s=1$.
First, we give our complexity result in the case of list-decoding 
via the Guruswami-Sudan algorithm~\cite{GurSud99}; then we show how
the re-encoding technique~\cite{KoeVar03b,KoMaVa11} can be 
used in our setting; then, we discuss the interpolation step of
the Wu algorithm~\cite{Wu08}; and finally
we present the application of our results
to the interpolation step of the soft-decoding~\cite{KoeVar03a}.
Note that in this last context, the $x_i$ in the input of
Problem~\ref{pbm:multivariate_interpolation} are not necessarily
pairwise distinct; we will nevertheless explain how to adapt our
algorithms to this case.
In these contexts of applications, we will use some of the
assumptions on the parameters \hyp{1}, \hyp{2}, \hyp{3},
\hyp{4} given in Section~\ref{sec:introduction}.

	\subsection{Interpolation step of the Guruswami-Sudan algorithm}
	\label{subsec:gursud}
We study here the specific context of the interpolation step of 
the Guruswami-Sudan list-decoding algorithm for Reed-Solomon codes.
This interpolation step is precisely Problem~\ref{pbm:multivariate_interpolation}
where we have $s=1$ and we make assumptions \hyp{1}, \hyp{2}, \hyp{3},
and \hyp{4}. Under \hyp{2}, the set $\Gamma$ introduced in 
Theorem~\ref{thm:complexity} reduces to
$\{j \in \Znn \,:\, j \leqslant \ell\} = \{0,\ldots,\ell\}$, so that 
$|\Gamma|=\ell+1$.
Thus, assumption~\hyp{1} ensures that the parameter $\varrho$
in that theorem is $\varrho=\ell+1$; because of \hyp{4} all
multiplicities are equal so that we further have
$M = \binom{m+1}{2} n = \frac{m(m+1)}2 n$.
From Theorem~\ref{thm:complexity}, we obtain the following result, 
which substantiates our claimed cost bound in 
Section~\ref{sec:introduction}, Table~\ref{table:prev}.

\begin{cor}
  \label{cor:complexity_gursud} 
  Taking $s=1$, if the parameters $\ell,n$, $m:=m_1=\cdots=m_n,b$ and $k:=k_1$ 
  satisfy \hyp{1}, \hyp{2}, \hyp{3}, and \hyp{4},
	then there exists a probabilistic algorithm that
  computes a solution to Problem~\ref{pbm:multivariate_interpolation}
  using \[ \mathcal{O}\big(\ell^{\omega-1} \M(m^2n)\log(mn)^2\big)
  \subseteq \mathcal{O}\tilde{~}(\ell^{\omega-1} m^2n) \]
  operations in $\K$, with
  probability of success at least $1/2$.
\end{cor}

We note that the probability analysis in Theorem~\ref{thm:complexity}
is simplified in this context.
Indeed, to ensure probability of success at least $1/2$, the algorithm chooses
$\mathcal{O}(m^2n)$ elements uniformly at random in a set $S \subseteq \K$ 
of cardinality at least $24m^4n^2$;
if $|\K| < 24m^4n^2$, one can use the remarks following 
Theorem~\ref{thm:complexity} in Section~\ref{sec:introduction} 
about solving the problem over an extension of $\K$ and retrieving a 
solution over $\K$. Here, the base field $\K$ of a Reed-Solomon code
must be of cardinality at least $n$ since the $x_i$ are distinct;
then, an extension degree $d = \mathcal{O}(\log_n (m))$ suffices
and the cost bound above becomes 
$\mathcal{O}\big(\ell^{\omega-1} \M(m^2n)\log(mn)^2 \cdot \M(d)\log(d)\big)$.
Besides, in the list-decoding of Reed-Solomon codes we have $m=\mathcal{O}(n^2)$, 
so that $d = \mathcal{O}(1)$ and the cost bound and probability of success 
in Corollary~\ref{cor:complexity_gursud} hold for \emph{any} field $\K$
(of cardinality at least $n$). 

	\subsection{Re-encoding technique}
	\label{subsec:reenc}

The re-encoding technique has been introduced by R. Koetter and 
A. Vardy~\cite{KoeVar03b,KoMaVa11} in order to reduce the cost of the 
interpolation step in list- and soft-decoding of Reed-Solomon codes.
Here, for the sake of clarity, we present this technique only 
in the context of Reed-Solomon list-decoding via the 
Guruswami-Sudan algorithm, using the same
notation and assumptions as in Subsection~\ref{subsec:gursud} above: 
$s=1$ and we have \hyp{1}, \hyp{2}, \hyp{3} and \hyp{4}.
Under some additional assumption on the input points in 
Problem~\ref{pbm:multivariate_interpolation}, by means 
of partially pre-solving the problem one obtains an 
interpolation problem whose linearization has smaller
dimensions. The idea at the core of this technique 
is summarized in the following lemma~\cite[Lemma 4]{KoeVar03b}.
\begin{lem}
	\label{lem:reenc_core}
	Let $m$ be a positive integer, $x$ be an element in $\K$ and 
	$Q = \sum_j Q_j(X) Y^j$ be a polynomial in $\K[X,Y]$.
	Then, 
	$Q(x,0) = 0$ with multiplicity at least $m$
	if and only if
	$(X-x)^{m-j}$ divides $Q_j$ for each $j<m$.
\end{lem}
\begin{proof}
	By definition,
	$Q(x,0) = 0$ with multiplicity at least $m$
	if and only if $Q(X+x,Y)$ has no monomial
	of total degree less than $m$.
	Since $Q(X+x,Y) = \sum_j Q_j(X+x) Y^j$, 
	this is equivalent to the fact that 
	$X^{m-j}$ divides $Q_j(X+x)$ for each
	$j<m$.
\end{proof}

This property can be generalized to the case of several roots
of the form $(x,0)$.
More precisely, the re-encoding technique is based on a shift
of the received word by a well-chosen code word, which 
allows us to ensure the following assumption on the points 
$\{(x_r,y_r)\}_{1\le r\le n}$:
\begin{equation}
	\label{eqn:assumption_reencoding}
	\text{for some integer } n_0 \ge k+1, \quad 
	y_1 = \cdots = y_{n_0} = 0 
	\;\;\text{ and }\;\; y_{n_0+1} \neq 0, \ldots, y_n \neq 0.
\end{equation}
We now define the polynomial $G_0 = \prod_{1\le r\le n_0} (X-x_r)$
which vanishes at $x_i$ when $y_i=0$, and 
Lemma~\ref{lem:reenc_core} can be rewritten as follows:
$Q(x_r,0) = 0$ with multiplicity at least $m$ for $1\le r\le n_0$
if and only if
$G_0^{m-j}$ divides $Q_j$ for each $j<m$.
Thus, we know how to solve the vanishing condition for the
$n_0$ points for which $y_r=0$: by setting each of the $m$ polynomials
$Q_0,\ldots,Q_{m-1}$ as the product of a power of $G_0$ and an
unknown polynomial. Combining this with the polynomial approximation
problem corresponding to the points $\{(x_r,y_r)\}_{n_0+1\le r\le n}$,
there remains to solve a smaller approximation problem.

Indeed, under the previously mentioned assumptions $s=1$,
\hyp{1}, \hyp{2}, \hyp{3} and \hyp{4}, it has been shown in
Section~\ref{sec:red} that the vanishing condition~\eqref{cond:iv} of 
Problem~\ref{pbm:multivariate_interpolation} restricted to points
$\{(x_r,y_r)\}_{n_0+1\le r\le n}$ is equivalent 
to the simultaneous polynomial approximations
\begin{equation}
	\label{eqn:poly_approx_gsa}
	\text{for } i<m, \;\; \quad\sum_{i\le j\le \ell} \binom{j}{i} R^{j-i}  Q_j = 0 \;\;\mod G^{m-i},
\end{equation}
where $G = \prod_{n_0+1\le r\le n} (X-x_r)$ and $R$ is the interpolation 
polynomial such that $\deg R < n-n_0$ and $R(x_r) = y_r$ 
for $n_0+1\le r\le n$. On the other hand, we have seen that the vanishing 
condition for the points $\{(x_r,y_r)\}_{1\le r\le n_0}$ is equivalent 
to $Q_{j} = G_0^{m-j} Q^{\star}_j$ for each $j<m$, 
for some unknown polynomials $Q^{\star}_0,\ldots,Q^{\star}_{m-1}$. 
Combining both equivalences, we obtain
\begin{equation}
	\label{eqn:poly_approx_reenc}
	\text{for } i < m, \;\; \quad\sum_{i\le j< m} F_{i,j}\, Q^{\star}_j 
	\;\;+\;\; \sum_{m\le j\le\ell} F_{i,j}\, Q_j 
	\;\;=\;\; 0 \mod G^{m-i},
\end{equation}
where for $i<m$,
\begin{equation}
	\label{eqn:reenc_fij}
	F_{i,j} = \binom{j}{i} R^{j-i} G_0^{m-j} \;\text{ for } i\le j<m \quad\text{and}\quad
	F_{i,j} = \binom{j}{i} R^{j-i} \;\text{ for } i\le j \le \ell.
\end{equation}

Obviously, the degree constraints on $Q_0,\ldots,Q_{m-1}$ directly
correspond to degree constraints on $Q^\star_0,\ldots,Q^\star_{m-1}$
while those on $Q_m,\ldots,Q_\ell$ are unchanged.
The number of equations in linearizations
of~\eqref{eqn:poly_approx_reenc} is $M' = \sum_{i<m} \deg(G^{m-i})
= \frac{m(m+1)}{2}(n-n_0)$, while the number of unknowns is
$N' = \sum_{j<m} (b-jk-(m-j)n_0) + \sum_{m\le j\le\ell} (b-jk) = 
\sum_{0\le j\le\ell} (b-jk)  - \frac{m(m+1)}{2}n_0$.
In other words, we have reduced the number of (linear) unknowns
as well as the number of (linear) equations by the same quantity
$\frac{m(m+1)}{2}n_0$,
which is the number of linear equations used to express the 
vanishing condition for the $n_0$ points $(x_1,0),\ldots,(x_{n_0},0)$.
(Note that if we were in the more general context of possibly distinct
multiplicities, we would have set $y_i = 0$ for the $n_0$ points which have 
the highest multiplicities, in order to maximize the benefit of the 
re-encoding technique.)

This re-encoding technique is summarized in Algorithm~\ref{algo:reencoding}.

\begin{center}
\fbox{\begin{minipage}{14.8cm}
  \begin{algo} \textbf{Interpolation step of list-decoding Reed-Solomon codes with re-encoding technique.}
  \label{algo:reencoding}

\medskip
  \noindent
  \emph{Input:} 
  $\ell,n,m,b,k$ in $\Zp$ such that 
  \hyp{1}, \hyp{2}, \hyp{3} and \hyp{4},
  and points $\{(x_r,y_r)\}_{1\le r\le n}$ \\ 
  \hspace*{1.15cm} in $\K^2$ with the $x_r$ pairwise 
  distinct and the $y_r$ as in~\eqref{eqn:assumption_reencoding}.

  \medskip
  \noindent
  \emph{Output:} $Q_0,\ldots,Q_\ell$ in $\K[X]$ such that $\sum_{j\le\ell} Q_jY^j$ is
  a solution to Problem~\ref{pbm:multivariate_interpolation} with
  \hspace*{1.56cm} input $s=1$,  $\ell,n,m=m_1=\cdots=m_n,b,k$ and $\{(x_r,y_r)\}_{1\le r\le n}$ 

  \begin{enumerate}[{\bf 1.}]
	  \item Compute $\mu=m, \nu=\ell+1, M_i'=(m-i)(n-n_0)$,
		  $N_j' = b-jk-n_0(m-j)$ for $j<m$ and $N_j' = b - jk$ for $j\ge m$
	  \item Compute $P_i = \left(\prod_{n_0+1\le r\le n} (X-x_r)\right)^{m-i}$ 
		  for $i<m$
	\item Compute the $F_{i,j}$ (modulo $P_i$) 
		  for $i<m,j\le \ell$ as in~\eqref{eqn:reenc_fij}
	\item Compute a solution $Q_0,\ldots,Q_\ell$ to Problem~\ref{pbm:poly_approx}
		  on input $\mu, \nu$, $M'_0,\ldots,M'_{m-1}, N'_0, \ldots,N'_{\ell}$
			and the polynomials $\{(P_i,F_{i,0},\ldots,F_{i,\ell})\}_{0 \le i < m}$
		\item Return $G_0^m Q_0, \; G_0^{m-1}Q_1, \,\ldots\,, \; G_0Q_{m-1},\; Q_m, \,\ldots\, , \;Q_\ell$
			(or report ``no solution'' if previous step did)
  \end{enumerate}
\end{algo}
\end{minipage}}
\end{center}

Assuming that Step~\textbf{4} is done using 
Algorithm~\ref{algo:interpolation_eke} or~\ref{algo:interpolation_direct}, 
we obtain the following result about list-decoding of Reed-Solomon 
codes using the re-encoding technique. 
\begin{cor}
  \label{cor:complexity_reenc} 
  Take $s=1$ and assume the parameters $\ell,n$, $m:=m_1=\cdots=m_n,b$ 
  and $k:=k_1$ satisfy \hyp{1}, \hyp{2}, \hyp{3} and \hyp{4}.
  Assume further that the points $\{(x_r,y_r)\}_{1\le r\le n}$ 
  satisfy~\eqref{eqn:assumption_reencoding} for some $n_0\ge k+1$. 
  Then there exists a probabilistic algorithm that
  computes a solution to Problem~\ref{pbm:multivariate_interpolation}
  using 
  \begin{align*}
	  \mathcal{O}\big(\ell^{\omega-1} \M(m^2 (n-n_0))\log(n-n_0)^2 +
	  m\M(mn_0) &+ \M(n_0)\log(n_0) \big) \\[2mm]
			 &\subseteq \mathcal{O}\tilde{~}(\ell^{\omega-1} m^2(n-n_0) + m^2n_0 )
  \end{align*}
  operations in $\K$
  with probability of success at least $1/2$.
\end{cor}
\begin{proof}
For steps~\textbf{1} to \textbf{3}, the complexity analysis is similar 
to the one in the proof of Proposition~\ref{prop:reduction}; we still note
that we have to compute $G_0$, so that these steps use 
$\mathcal{O}(\ell \M(m^2(n-n_0))\log(n-n_0) + \M(n_0)\log(n_0))$ operations in $\K$.
According to Theorem~\ref{thm:poly_approx_complexity},
Step~\textbf{4} uses 
$\mathcal{O}\big(\ell^{\omega-1} \M(m^2 (n-n_0))\log(n-n_0)^2\big)$ 
operations in $\K$.
Step~\textbf{5} uses 
$\mathcal{O}(m\M(mn_0) + \M(m^2(n-n_0)))$ operations in~$\K$. 
Indeed, we first compute $G_0,\ldots,G_0^m$ using 
$\mathcal{O}(m \M(mn_0))$ operations and then 
the products $G_0^{m-j} Q_j$ for $j<m$ are computed using 
$\mathcal{O}(m\M(mn_0) + \M(m^2(n-n_0)))$ operations:
for each $j<m$, the product $G_0^{m-j} Q_j$
can be computed using $\mathcal{O}(\M(mn_0) + \M(\deg(Q_j)))$ operations
since $G_0^{m-j}$ has degree at most $mn_0$; and from 
Algorithms~\ref{algo:interpolation_eke} and~\ref{algo:interpolation_direct}
we know that $\deg Q_0 + \cdots + \deg Q_{m-1} \le (\sum_{i<m} M'_i)+1$
(see~\eqref{eqn:quasi-square_system} in Section~\ref{subsec:key_equations_version}),
with here $\sum_{i<m} M'_i = \frac{m(m+1)}{2}(n-n_0)$.
\end{proof}

Similarly to the remarks following Corollary~\ref{cor:complexity_gursud}, 
if $|\K| < 24m^2(n-n_0)$ then $\K$ does not contain enough elements 
to ensure a probability of success at least $1/2$ using our algorithms, 
but one can solve the problem over an extension of degree~$\mathcal{O}(1)$ 
and retrieve a solution over $\K$ without impacting the cost bound.

	\subsection{Interpolation step in the Wu algorithm}

Our goal now is to show that our
algorithms can also be used to efficiently solve the 
interpolation step in the Wu algorithm.
In this context, we have $s=1$ and 
we make assumptions \hyp{1}, \hyp{2} and \hyp{4}
on input parameters to 
Problem~\ref{pbm:multivariate_interpolation}. 
We note that here the weight $k$ is no longer linked to 
the dimension of the code; besides, we may have $k\le 0$.

Roughly, the Wu algorithm~\cite{Wu08} works as follows.
It first uses the Berlekamp-Massey algorithm to reduce the problem 
of list-decoding a Reed-Solomon code to a problem of rational 
reconstruction which focuses on the error locations
(while the Guruswami-Sudan algorithm directly relies on a problem
of polynomial reconstruction which focuses on the correct locations).
Then, it solves this problem using an interpolation step 
and a root-finding step which are very similar to
the ones in the Guruswami-Sudan algorithm.

Here we focus on the interpolation step, which differs from
the one in the Guruswami-Sudan algorithm by mainly one
feature: the points $\{(x_r,y_r)\}_{1\le r\le n}$ lie
in $\K \times (\K\cup\{\infty\})$, that is,
some $y_r$ may take the special value $\infty$.
For a point $(x,\infty)$, a polynomial $Q$ 
in $\K[X,Y]$ and a parameter $\ell$
such that $\deg_Y(Q) \le\ell$,
Wu defines in~\cite{Wu08} the vanishing condition 
$Q(x,\infty) = 0$ with multiplicity at least $m$ as 
the vanishing condition $\overline{Q}(x,0) = 0$ with
multiplicity at least $m$, where 
$\overline{Q} = Y^\ell Q(X,Y^{-1})$ is the 
reversal of $Q$ with respect to the variable $Y$
and the parameter $\ell$.
Thus, we have the following direct adaptation of
Lemma~\ref{lem:reenc_core}.

\begin{lem}
	\label{lem:Wu_infty_point}
	Let $\ell,m$ be positive integers, $x$ be an element in $\K$ and 
	$Q = \sum_{j\le\ell} Q_j(X) Y^j$ be a polynomial in $\K[X,Y]$ with
	$\deg_Y(Q)\le\ell$.
	Then, 
	$Q(x,\infty) = 0$ with multiplicity at least $m$
	if and only if
	$(X-x)^{m-j}$ divides $Q_{\ell-j}$ for each $j<m$.
\end{lem}

As in the re-encoding technique, assuming we reorder the points so that
$y_1=\cdots=y_{n_\infty} = \infty$ and $y_r \neq \infty$ for $r>n_\infty$
for some $n_\infty \ge 0$, the vanishing condition
of Problem~\ref{pbm:multivariate_interpolation} restricted
to the points $\{(x_r,y_r)\}_{1\le r\le n_\infty}$ is equivalent
to $Q_{\ell-j} = G_\infty^{m-j} Q^\star_{\ell-j}$ for each $j<m$,
for some unknown polynomials $Q^\star_{\ell-m+1},\ldots,Q^\star_\ell$.
The degree constraints on $Q_{\ell-m+1},\ldots,Q_{\ell}$ directly
correspond to degree constraints on $Q^\star_{\ell-m+1},\ldots,Q^\star_{\ell}$,
while those of $Q_0,\ldots,Q_{\ell-m}$ are unchanged.

This means that in the interpolation problem we are faced with,
we can deal with the points of the form $(x,\infty)$ the same way
we dealt with the points of the form $(x,0)$ in the case of 
the re-encoding technique: we can pre-solve the corresponding
equations efficiently, and we are left with an approximation problem 
whose dimensions are smaller than if no special attention had
been paid when dealing with the points of the form $(x,\infty)$. 
More precisely, 
defining
$G_\infty = \prod_{1\le r\le n_\infty} (X-x_r)$ as well as
$G = \prod_{n_\infty +1\le r\le n} (X-x_r)$ and 
$R$ of degree less than $n-n_\infty$ such that $R(x_r) = y_r$
for each $r>n_\infty$, the polynomial approximation problem 
we obtain is
\begin{equation}
	\label{eqn:poly_approx_wu}
	\text{for } i < m, \;\; \sum_{i\le j\le \ell -m} F_{i,j}\,  Q_j 
	\;\;+\;\; \sum_{\ell - m < j \le\ell} F_{i,j}\, Q^{\star}_j
	\;\;=\;\; 0 \mod G^{m-i}
\end{equation}
where for $i<m$,
\begin{equation}
	\label{eqn:wu_fij}
	F_{i,j} = \binom{j}{i} R^{j-i} \;\text{ for } i\le j\le\ell-m \quad\text{and}\quad
	F_{i,j} = \binom{j}{i} R^{j-i} G_\infty^{j-\ell+m} \;\text{ for } \ell - m < j \le\ell.
\end{equation}
Pre-solving the equations for the points of the form $(x,\infty)$
has led to reduce the number of (linear) unknowns
as well as the number of (linear) equations by the same quantity
$\frac{m(m+1)}{2}n_\infty$,
which is the number of linear equations used to express the 
vanishing condition for the $n_\infty$ points $(x_1,\infty),\ldots,(x_{n_\infty},\infty)$.
We have the following result.

\begin{cor}
  \label{cor:complexity_wu} 
  Take $s=1$ and assume the parameters $\ell,n$, $m:=m_1=\cdots=m_n,b$ 
  and $k:=k_1$ satisfy \hyp{1}, \hyp{2} and \hyp{4}.
  Assume further that each the points $\{(x_r,y_r)\}_{1\le r\le n}$ 
  is allowed to have the special value $y_r = \infty$.
  
  Then there exists a probabilistic algorithm that
  computes a solution to Problem~\ref{pbm:multivariate_interpolation}
  using \[ \mathcal{O}\big(\ell^{\omega-1} \M(m^2 n)\log(n)^2\big) 
  \subseteq \mathcal{O}\tilde{~}(\ell^{\omega-1} m^2n ) \]
  operations in $\K$
  with probability of success is at least $1/2$.
\end{cor}

	As above, if $|\K| < 24m^2(n-n_\infty)$ 
	then in order to ensure a probability of success at least $1/2$ 
	using our algorithms, one can solve the problem over an extension 
	of degree~$\mathcal{O}(1)$ and retrieve a solution over $\K$,
	without impacting the cost bound.
	
	We note that unlike in the re-encoding technique where the 
	focus was on a reduced cost involving $n-n_0$, here we are
	not interested in writing the detailed cost involving
	$n-n_\infty$. The reason is that $n_\infty$ is expected to
	be close to $0$ in practice. The main advantage of the Wu
	algorithm over the Guruswami-Sudan algorithm is that it 
	uses a smaller multiplicity $m$, at least for practical 
	code parameters; details about the choice of parameters 
	$m$ and $\ell$ in the context of the Wu algorithm can be 
	found in~\cite[Section IV.C]{BeHoNiWu13}.

	\subsection{Application to soft-decoding of Reed-Solomon codes}
	\label{subsec:softdec}

As a last application, we briefly sketch how to adapt our results to the
context of soft-decoding, in which we still have $s=1$. 
The interpolation step in soft-decoding
of Reed-Solomon codes~\cite{KoeVar03a} differs from
Problem~\ref{pbm:multivariate_interpolation} because there is no
assumption ensuring that the $x_r$ are pairwise distinct among 
the points $\{(x_r,y_r)\}_{1\le r\le n}$.
Regarding our algorithms, this is not a minor issue since 
this assumption is at the core of the reduction in Section~\ref{sec:red};
we will see that we can still rely on Problem~\ref{pbm:poly_approx}
in this context. 
However, although the number of linear equations 
$\sum_{1\le r\le n} \frac{m_r(m_r+1)}{2}$ imposed by the vanishing 
condition is not changed by the fact that several $x_r$ can be the same 
field element, it is expected that the reduction to 
Problem~\ref{pbm:poly_approx} will not be as 
effective as if the $x_r$ were pairwise distinct. More precisely,
the displacement rank of the structured matrix in the 
linearizations of the problem in 
Algorithms~\ref{algo:interpolation_eke} and~\ref{algo:interpolation_direct}
may in some cases be larger than if the $x_r$ were pairwise distinct.

To measure to which extent we are far from the situation
where the $x_r$ are pairwise distinct, we use the parameter
\[q = \max_{x\in\K} \big|\{r \in\{1,\ldots,n\} \;\mid\; x_r = x\}\big|\;.\]
For example, $q =1$ corresponds to pairwise distinct $x_r$
while $q=n$ corresponds to $x_1=\cdots=x_n$; we always have
$q\le n$ and, if $\K$ is a finite field, $q \le |\K|^s$ with $s=1$
in our context here. 
Then, we can write the set of points $\PP = \{(x_r,y_r)\}_{1\le r\le n}$
as the disjoint union of $q$ sets $\PP = \PP_1 \cup \cdots \cup \PP_q$ 
where each set $\PP_h = \{(x_{h,r},y_{h,r})\}_{1\le r\le n_h}$
is such that the $x_{h,r}$ are pairwise distinct; we denote
$m_{h,r}$ the multiplicity associated to the point $(x_{h,r},y_{h,r})$
in the input of Problem~\ref{pbm:multivariate_interpolation}.
Now, the vanishing condition~\eqref{cond:iv} asks that the $q$ 
vanishing conditions restricted to each $\PP_h$ hold simultaneously.
Indeed, $Q(x_r,y_r) = 0$ with multiplicity at least
$m_r$ for all points $(x_r,y_r)$ in $\PP$ if and only if 
for each set $\PP_h$, $Q(x_{h,r},y_{h,r})=0$ with
multiplicity at least $m_{h,r}$ for all points $(x_{h,r},y_{h,r})$
in $\PP_h$.

We have seen in Section~\ref{sec:red} how to rewrite the vanishing
condition as simultaneous polynomial approximations when the $x_r$
are pairwise distinct. This reduction extends to this case:
by simultaneously rewriting the vanishing condition
for each set $\PP_h$, one obtains a problem of simultaneous 
polynomial approximations whose solutions exactly correspond
to the solutions of the instance of (extended)
Problem~\ref{pbm:multivariate_interpolation} we are considering.
Here, we do not give details about this reduction; they can be found 
in~\cite[Section 5.1.1]{Zeh13}.
Now, let $m^{(h)}$ be the largest multiplicity among those of the points
in $\PP_h$; in this reduction to Problem~\ref{pbm:poly_approx}, 
the number of polynomial equations we obtain is 
$\sum_{1\le h\le q} m^{(h)}$. Thus, according to 
Theorem~\ref{thm:poly_approx_complexity}, for solving this instance 
of Problem~\ref{pbm:poly_approx}, our Algorithms~\ref{algo:interpolation_eke} 
and~\ref{algo:interpolation_direct} use $\softO{\rho^{\omega-1} M'}$
operations in $\K$, where $\rho = \max(\ell+1,\sum_{1\le h\le q} m^{(h)})$
and $M' =\sum_{1\le r\le n} \frac{m_r(m_r+1)}{2}$. 
We see in this cost bound that the distribution of the points into
disjoint sets $\PP = \PP_1 \cup \cdots \cup \PP_q$
has an impact on the number of polynomial equations in
the instance of Problem~\ref{pbm:poly_approx} we get:
when choosing this distribution, multiplicities could 
be taken into account in order to minimize this impact.


\appendix
\makeatletter
\def\@seccntformat#1{Appendix\ \csname the#1\endcsname .\quad}
\makeatother

\section{On assumption \hyp{1}}
\label{app:H1}

In this appendix, we discuss the relevance of assumption~\hyp{1}
that was introduced previously for  
Problem~\ref{pbm:multivariate_interpolation}. In the introduction, 
we did not make any assumption on $m = \max_{1\le i\le n} m_i$ and
$\ell$, but we mentioned that the assumption \hyp{1}, that is, $m \le
\ell$ is mostly harmless. The following lemma substantiates this
claim, by showing that the case $m > \ell$ can be reduced to the case
$m = \ell$. 

\begin{lem}\label{lem:mell}
  Let $s,\ell,n,m_1,\ldots,m_n,b,\bk$ be parameters for 
  Problem~\ref{pbm:multivariate_interpolation}, and suppose 
  that $m > \ell$. Define 
  $P = \prod_{1\le i\le n: \; m_i > \ell} (X-x_i)^{m_i-\ell}$ 
  and $d=\deg(P)$. 
  The solutions to this problem are the polynomials 
  of the form $Q= Q^\star\, P$ with $Q^\star$ a solution
  for the parameters $s,\ell,n,m'_1,\ldots,m'_n,b-d,\bk$,
  where $m'_i = \ell$ if $m_i>\ell$ and $m'_i = m_i$ otherwise.
\end{lem}
\begin{proof}
 Assume a solution exists, say $Q$, and let 
 $Q_i(X,\bY) = Q(X+x_i, Y_1+y_{i,1}, \ldots, Y_s+y_{i,s})$ for $i=1,\ldots,n$.
 Every monomial of $Q_i$ has the form $X^h\bY^\bj$ with $h \ge m_i-\ell$, 
 since $|\bj|\le\ell$ by condition~\eqref{cond:ii}
 and $h+|\bj| \ge m_i$ by condition~\eqref{cond:iv}.
 Therefore, if $m_i>\ell$ then $X^{m_i-\ell}$ divides $Q_i$ and, 
 shifting back the coordinates for each $i$, we deduce that $P$ divides $Q$.

 Let us now consider the polynomial $Q^\star = Q/P$ and show that it solves 
 Problem~\ref{pbm:multivariate_interpolation} for the parameters 
 $s,\ell,n,m'_1,\ldots,m'_n,b-d,\bk$.
 First, $Q^\star$ clearly satisfies conditions~\eqref{cond:i} and~\eqref{cond:ii}.
 Furthermore, 
 writing $Q = \sum_\bj Q_\bj(X) \bY^\bj$ and $Q^\star = \sum_\bj Q_j^\star(X) \bY^\bj$,
 we have $Q_\bj^\star = Q_\bj / P$ for all $\bj$, so that
 \[\wdeg_\bk(Q^\star) \;=\; \max_\bj (\deg(Q_\bj) - d 
\:+\: k_1j_1 + \cdots + k_sj_s) \;=\; \wdeg_\bk(Q)-d \;<\; b-d \,,\]
 so that condition~\eqref{cond:iii} holds for $Q^\star$ with $b$ replaced by $b-d$.
 Finally, $Q^\star$ satisfies condition~\eqref{cond:iv} with the $m_i>\ell$ 
 replaced by $m'_i = \ell$:
 writing $Q_i^\star(X,\bY) = Q^\star(X+x_i, Y_1+y_{i,1}, \ldots, Y_s+y_{i,s})$ 
 for $i\in\{1,\ldots,n\}$ such that $m_i>\ell$, we have
 \[ Q_i^\star(X,\bY) = \frac{Q_i(X,\bY)}{X^{m_i-\ell} \, P_i(X)}, \qquad
	 \text{where}\quad P_i(X) = \prod_{h\ne i:\; m_h>\ell}(X+x_i-x_h)^{m_h-\ell};
 \]
 all the monomials of $Q_i(X,\bY) / X^{m_i-\ell}$ have the form $X^h\bY^\bj$
 with $h + |\bj| \ge m_i-(m_i-\ell) = \ell$ and, since $P_i(0) \ne 0$, 
 the same holds for $Q_i^\star(X,\bY)$.

Conversely, let $Q'$ be \emph{any} solution to
Problem~\ref{pbm:multivariate_interpolation} with parameters
$s,\ell,n,m'_1,\ldots,m'_n,b-d,\bk$. Proceeding as in the previous
paragraph, one easily verifies that the product $Q'\, P$ is a
solution to Problem~\ref{pbm:multivariate_interpolation} with 
parameters $s,\ell,n,m_1,\ldots,m_n,b,\bk$.
\end{proof}

\section{On assumption \hyp{3}}
\label{app:H5}

In this appendix, we show the relevance of the assumption
``$k_j < n$ for some $j\in\{1,\cdots,s\}$'' when considering 
Problem~\ref{pbm:multivariate_interpolation}; in particular
when $s=1$ or when we assume that $k_1 = \cdots = k_s =: k$, 
this shows the relevance of the assumption $\hyp{3}: k<n$. 
More precisely, when $k_j \geqslant n$ for every $j$, 
Lemma~\ref{lem:k_vs_n} below gives an explicit solution to 
Problem~\ref{pbm:multivariate_interpolation}.

\begin{lem} \label{lem:k_vs_n}
  Let $s,\ell,n,\bm,b,\bk$ be parameters for 
  Problem~\ref{pbm:multivariate_interpolation}
  and suppose that $k_j \ge n$ for $j=1,\ldots,s$.
  Define $P = \prod_{1\le i\le n} (X-x_i)^{m_i}$ and
  $d = \deg(P) = \sum_{1\le i\le n}m_i$.
  If $b \leqslant d$ then this problem has no solution. 
  Otherwise, a solution is given by the polynomial $P$
  (considered as an element of $\K[X,\bY]$).
\end{lem}
\begin{proof}
If $b > d$ then it is easily checked that $P$ 
satisfies conditions~\eqref{cond:i}--\eqref{cond:iv}
and thus solves Problem~\ref{pbm:multivariate_interpolation}.
Now, to conclude the proof, let us show that if 
Problem~\ref{pbm:multivariate_interpolation} admits a 
solution $Q$, then $b > d$ must hold.
Let $d_Y = \deg_\bY Q$. If $d_Y \geqslant m = \max_i m_i$, then the 
weighted-degree condition~\eqref{cond:iii} 
gives $b > \wdeg_\bk(Q)\ge d_Y (\min_j k_j) \ge mn \ge d$.
Let us finally assume $d_Y < m$. Following the proof of Lemma~\ref{lem:mell}, 
we can write $Q = P^\star \, Q^\star$ where 
$P^\star = \prod_{1\le i\le n: \; m_i > d_Y} (X-x_i)^{m_i-d_Y}$,
for some $Q^\star$ in $\K[X,\bY]$ such that $\deg_\bY Q^\star = d_Y$.
Then, the weighted-degree condition gives $b > 
\sum_{1\le i\le n:\;m_i>d_Y}(m_i-d_Y) + 
\wdeg_\bk(Q^\star) \ge \sum_{1\le i\le n:\;m_i>d_Y}(m_i-d_Y) + d_Yn
\ge \sum_{1\le i\le n:\;m_i>d_Y}m_i + \sum_{1\le i\le n:\; m_i\le d_Y} d_Y
\ge d$.
\end{proof}

\section{The lattice-based approach}
\label{app:lattice_comparison}

In this appendix, we summarize the approach for solving
Problem~\ref{pbm:multivariate_interpolation} via
the computation of a reduced polynomial lattice basis;
this helps us to compare the cost bounds for this approach 
with the cost bound we give in Theorem~\ref{thm:complexity}.
Here, $s\ge 1$ and for simplicity, we assume that 
$k := k_1 = \cdots = k_s$ as in the list-decoding
of folded Reed-Solomon codes. Besides, we make the 
assumptions~\hyp{1}, \hyp{2}, \hyp{3} and \hyp{4} 
as presented in the introduction.
Two main lattice constructions exist in the literature;
following~\cite[\S4.5]{PB08}, we
present them directly in the case $s \geqslant 1$, and then give the cost
bound that can be obtained using polynomial lattice reduction to find a
short vector in the lattice. 

Let $G=\prod_{1\le r\le n}(X-x_r)$ and $R_1,\ldots,R_s\in \K[X]$ 
such that $\deg(R_j) < n$ and $R_j(x_i) = y_{i,j}$,
for every $j\in\{1,\ldots,s\}$ and $i\in\{1,\ldots,n\}$.
In the first construction, the lattice is generated by
the polynomials
\begin{align*}
  & \left \{G^i\prod_{r = 1}^s (Y_r - R_r)^{j_r}\ \middle |\ i >0,\ j_1,\dots,j_s \geqslant 0,\ i+ \sbj = m\right \}\\
  \bigcup\quad 	& \left \{\prod_{r = 1}^s(Y_r - R_r)^{j_r} Y_r^{J_r}\ \middle |\ j_1,\dots,j_s \geqslant 0, \
  J_1,\dots,J_s \geqslant 0, \
  \sbj = m,\ \sbJ \le   \ell-m \right \};
\end{align*}
this construction may be called \emph{banded} 
due to the shape of the generators above when $s=1$. 
In the second construction, which may be called \emph{triangular},
the lattice is generated by the polynomials
\begin{align*}
  & \left \{G^i\prod_{r = 1}^s (Y_r - R_{r})^{j_r}\ \middle |\ i>0,\ j_1,\dots,j_s \geqslant 0,\ i+ \sbj = m \right \}\\
  \bigcup\quad & \left  \{\prod_{r = 1}^s(Y_r - R_{r})^{j_r}\ \middle |\ j_1,\dots,j_s\geqslant 0,\ m \le   \sbj \le   \ell \right \}.
\end{align*}

When $s=1$, the first construction is used
in~\cite[Remark~16]{BeeBra10} and~\cite{LeeOSul08,CohHen11a}, and the
second one is used in~\cite{BeeBra10,Bernstein11}; when $s\geqslant 1$,
the former can be found in~\cite{PB08} while the latter appears 
in~\cite{KB10,CohHen12}.  In both cases the actual lattice
bases are the coefficient vectors (in~$\bY$) of the polynomials
$h(X,X^k Y_1,\dots,X^k Y_s)$, for $h$ in either of the sets above;
these $X^k$ are introduced to account for the weighted-degree
condition~(\ref{cond:iii}) in Problem~\ref{pbm:multivariate_interpolation}.
		
In this context, for a lattice of dimension $L$ given by generators 
of degree at most $d$, the algorithm in~\cite{GiJeVi03} computes 
a shortest vector in the lattice in expected time 
$\mathcal{O}(L^\omega \M(d) \log(Ld))$, as detailed below.
For a deterministic solution, one can use the
algorithm of Gupta, Sarkar, Storjohann, and Valeriote~\cite{GuSaStVa12},
whose cost is in $\mathcal{O}(L^\omega \M(d) ((\log(L))^2 + \log(d)) )$.)

For the banded basis, its dimension $L_B$ and degree $d_B$ can be taken as follows:
\[L_B = {{s+m-1} \choose {s}}+ {{s+m-1} \choose {s-1}}{{s+\ell-m} \choose {s}} 
\quad\text{and}\quad d_B = \mathcal{O}(mn).\] 
The dimension formula is given explicitly in~\cite[p.~75]{PB08},
while the degree bound is easily obtained when assuming that the 
parameters $m,n,b$ of Problem~\ref{pbm:multivariate_interpolation}
satisfy $b \leqslant mn$; such an assumption is not restrictive, 
since when $b>mn$ the polynomial $Q = G^m$ is a trivial solution.
In this case, the arithmetic cost for constructing the lattice matrix
with the given generators is 
$\mathcal{O}\left(\binom{s+m}{s}^2 \,\M(mn) \right)$, which is
$\mathcal{O}(L_B^2 \,\M(mn))$.
Similarly, in the triangular case,
\[L_T = {{s+\ell}\choose {s}} \quad\text{and}\quad d_T = \mathcal{O}(\ell n),\]
and the cost for constructing the lattice matrix is
$\mathcal{O}(L_T^2 \,\M(\ell n)) $.

Under our assumption $\hyp{1}:  m \le \ell$, we always have $L_B \geqslant L_T$
and $d_B \le d_T$; when $s=1$, we get $L_B=L_T = \ell+1$.  

To bound the cost of reducing these two polynomial lattice bases, 
recall that the algorithm of~\cite{GiJeVi03} works as follows.
Given a basis of a lattice of dimension $L$ and degree $d$,
if $x_0\in\K$ is given such that the determinant of the lattice 
does not vanish at $X=x_0$, then the basis will be reduced 
deterministically using $\mathcal{O}(L^\omega \M(d) \log(Ld))$
operations in $\K$.
Otherwise, such an $x_0$ is picked at random in $\K$ or, if 
the cardinality $|\K|$ is too small to ensure success
with probability at least $1/2$, in a field extension $\extK$ of $\K$.
In general, $\extK$ should be taken of degree $\mathcal{O}(\log(Ld))$
over $\K$; however, here degree 2 will suffice.
Indeed, following~\cite[p.~206]{Bernstein11} 
we note that for the two lattice constructions above
the determinants have the special form $G(X)^{i_1}X^{i_2}$ 
for some $i_1,i_2\in\Znn$.
Since $G(X) = (X-x_1)\cdots (X-x_n)$ with $x_1,\ldots,x_n \in \K$ pairwise distinct, 
$x_0$ can be found deterministically in time $\mathcal{O}(\M(n)\log(n))$ as soon as $|\K| > n+1$,
by evaluating $G$ at $n+1$ arbitrary elements of $\K$;
else, $|\K|$ is either $n$ or $n+1$, and $x_0$ can be found in an extension $\extK$ of $\K$ of degree $2$.
Such an extension can be computed with probability of success at least $1/2$ in time 
$\mathcal{O}(\log(n))$ (see for example~\cite[\S14.9]{vzGathen03}).
Then, with the algorithm of~\cite{GiJeVi03} we obtain a reduced
basis over $\extK[X]$ using $\mathcal{O}(L^\omega \M(d) \log(Ld))$ 
operations in $\extK$; since the degree of $\extK$ over $\K$ is $\mathcal{O}(1)$, this is $\mathcal{O}(L^\omega \M(d) \log(Ld))$ operations in $\K$.
Eventually, one can use~\cite[Theorems 13 and 20]{SarSto11} to transform
this basis into a reduced basis over $\K[X]$ without impacting 
the cost bound; or more directly, since here we are only looking for a
sufficiently short vector in the lattice, this vector can be extracted
from a shortest vector in the reduced basis over $\extK[X]$.
Therefore, by applying the algorithm of~\cite{GiJeVi03} to 
reduce the banded basis and triangular basis shown above, 
we will always obtain a polynomial $Q$ solution to 
Problem~\ref{pbm:multivariate_interpolation} (assuming one exists)
in expected time
\[ \mathcal{O}(L_B^\omega \M(mn) \log(L_B m n)) 
\quad\text{and}\quad 
\mathcal{O}(L_T^\omega \M(\ell n) \log(L_T \ell n)),
\] 
respectively.
For $s=1$, 
thanks to the assumption $\hyp{1}$, 
these costs become $\mathcal{O}(\ell^\omega \M(mn) \log(\ell n))$
and $\mathcal{O}(\ell^\omega \M(\ell n) \log(\ell n))$, respectively,
and are those reported in~\cite{CohHen11a,Bernstein11}.
For $s > 1$, the costs reported in~\cite{PB08,KB10} are worse, 
but only because the short vector algorithms used in those 
references are inferior to the ones we refer to; no cost bound is 
explicitly given in~\cite{CohHen12}.
The result in Theorem~\ref{thm:complexity} is an improvement
over those of both~\cite{PB08} and \cite{KB10}.
To see this, remark that the cost in our theorem is quasi-linear 
in ${{s+\ell}\choose {s}}^{\omega-1} {{s+m}\choose {s+1}} n$, 
whereas the costs in~\cite{PB08,KB10} are at least 
${{s+\ell}\choose {s}}^\omega m n$; a simplification proves our claim.

\bigskip \noindent{\bf Acknowledgments.} 
Muhammad F. I. Chowdhury and {\'E}ric Schost were supported
by NSERC and by the Canada Research Chairs program. 
Vincent Neiger is supported by the international mobility 
grant Explo'ra Doc from R\'egion Rh\^one-Alpes.
We thank the three reviewers for their helpful comments on 
a preliminary version of this work~\cite{CJNSV12},
and especially the second one
for suggesting a shorter proof of
Lemma~\ref{lem:fast_generator_computation}.

\bibliographystyle{plain} 
\bibliography{biblio}

\end{document}